\documentclass[12pt,draftcls,onecolumn]{IEEEtran}
\ifCLASSINFOpdf
\else
\fi
\usepackage{epsfig}
\usepackage{amssymb}
\usepackage{amsmath}
\usepackage{amsfonts}
\newtheorem{theorem}{Theorem}
\newtheorem{proposition}{Proposition}

\newtheorem{lemma}{Lemma}
\newtheorem{corollary}{Corollary}

\newtheorem{remark}{Remark}

\newtheorem{assumption}{Assumption}
\newtheorem{proof}{Proof}
\newtheorem{proof of Theorem 1}{Proof of Theorem 1}

\begin{document}
%

\title{Distributed algorithms for solving a class of convex feasibility problems}
%
%
%

\author{Kaihong~Lu,~
       Gangshan~Jing,~
       and~Long~Wang~
\thanks{This work was supported by National Science Foundation of China (Grant Nos. 61533001 61375120, and 61603288). (\emph{Corresponding author: Long Wang})}
\thanks{K. Lu and G. Jing are with the Center for Complex Systems, School of Mechano-electronic Engineering, Xidian University, Xi'an 710071, China
 (e-mail:khong\_lu@163.com; nameisjing@gmail.com)
 }
\thanks{L. Wang is with the Center for Systems and Control, College of Engineering, Peking
University, Beijing 100871, China
 (e-mail: longwang@pku.edu.cn)}
}

\maketitle

\begin{abstract}

In this paper, a class of convex feasibility problems (CFPs) are studied for multi-agent systems through
local interactions. The objective is to search a feasible solution to the convex inequalities with some set constraints in a distributed manner. The
distributed control algorithms, involving subgradient and projection, are proposed for both continuous-
and discrete-time systems, respectively. Conditions associated with connectivity of the directed communication graph are given to ensure convergence of the algorithms. It is shown that under mild conditions, the states of all agents reach consensus asymptotically and the consensus state is located in the solution set of the CFP. Simulation examples are presented to demonstrate the effectiveness of the theoretical results.
\end{abstract}

\begin{keywords}
Multi-agent systems; Consensus; Convex inequalities; Subgradient; Projection.
\end{keywords}

%
\IEEEpeerreviewmaketitle
\section{Introduction}
\IEEEPARstart{D}{istributed} coordination control of multi-agent systems (MASs) has been intensively investigated in various areas including engineering, natural science, and social science \cite{Y. Guan1}-\cite{MaJ3}. As a fundamental coordination problem, the consensus which requires that a group of autonomous agents achieve a common state has attracted much attention, see \cite{Jing4}-\cite{Y. Tang 1}. This is due to its wide applications in distributed control and estimation \cite{S. Kar9}, distributed optimization \cite{Lin11}-\cite{Rahili13} and distributed methods for solving linear equations \cite{S. Mou15, H. Cao17}.

Researches on consensus can be roughly categorized depending on whether the agents have continuous- or discrete- time dynamics. Noticeable works focusing on the multi-agent systems include \cite{Olfati-Saber5, F. Xiao18, Martin22, Hendrickx23} for the
continuous-time case and \cite{W. Ren6}, \cite{Hendrickx23}-\cite{L. Moreau21} for the
discrete-time case. In the aforementioned works, the agents interact with each other through a network and each agent adjusts its own state by using only local information from its neighbors. Within this framework,
connectivity of the communication graph plays a key role in achieving consensus, and consequently several
 conditions of the connectivity have been established. For example, the communication graph must have
a spanning tree when the topology is fixed \cite{Olfati-Saber5}, while the union of the communication graphs
should have a spanning tree frequently enough as the system
evolves when the topology is switching \cite{F. Xiao18}, \cite{L. Moreau21}. In addition, infinitely-joint connectedness, i.e., the infinitely occurring communication graphs are jointly connected, is necessary to make the agents reach consensus when the topology is time-varying \cite{Martin22}, \cite{Hendrickx23}.

In recent years, the constrained consensus problem that seeks to reach state agreement in the intersection of a
number of convex sets has been widely investigated. In \cite{A. Nedi24}, a projection-based consensus algorithm was proposed when the communication graph is balanced. This algorithm with time delays was studied in \cite{P. Lin25}, where the union of the communication graphs within a period was assumed to be strongly connected. The problem was extended to the continuous-time case in \cite{G. Shi26}, where each set serves as an optimal solution set of a local objective function, and the global optimal solution is achieved as long as the intersection of the constrained sets is computed. By taking the advantages of the property that the solution set of linear equations is an affine set, the projection-based consensus algorithm in \cite{G. Shi26} was successfully applied to solving linear equations in \cite{G. Shi27}, where the projection operator in \cite{G. Shi26} was replaced with a special affine projection operator. Unlike the distributed algorithm for solving linear equations in \cite{S. Mou15}, the projection-based consensus algorithm in \cite{G. Shi27} does not need to restrict each agent's initial state within the solution set of its corresponding equations. The methods in \cite{A. Nedi24}-\cite{G. Shi27} are useful for the computation of the intersection when the projections onto the local sets are easily calculated. However, in general, the application of the projected
method usually requires the solution of an auxiliary minimization
problem associated with the projection onto the local set at each time. This might lead to a limitation on its applications.

Comparing with computing the intersection and solving linear equations, a more general problem is solving CFPs, which usually needs to solve linear equations and convex inequalities simultaneously, and ensure the solution to be in the intersection of some simple convex sets. Applications of solving CFPs arise in different fields, such as pattern recognition \cite{Cover T M38}, signal processing \cite{Yamada I39} and image restoration  \cite{Y. Censor28}, \cite{G. T. Herman29}. It is also well known that some convex programming problems can be transformed into an equivalent CFP through the Karush-Kuhn-Tucker condition \cite{Han S P}. For example, the linear program problem in \cite{D. Richert30} can be transformed into a set of linear equations and inequalities. Inspired by the distributed methods for solving linear equations \cite{S. Mou15, G. Shi27}, distribute methods for CFPs will be studied in this paper. Different from linear equations, the solution set of a CFP is usually not a simple affine set due to the existence of inequalities which can even be nonlinear, thus it is necessary to develop alternative methods for solving this problem.

In this paper, distributed algorithms, involving subgradient and projection, are proposed for multi-agent systems to solve the CFP involving convex inequalities. Here the distributed control algorithms are designed for the continuous- and the discrete-time systems, respectively. Our aim is to obtain the graphic criteria for the convergence of these algorithms. One of the challenge is that, the subgradient and projection operations lead to nonlinearity of the algorithms. To deal with this problem, the control inputs are decomposed into a linear part involving the traditional consensus term and a nonlinear part involving the subgradient and projection operations. The linear part is analyzed by using the graph theory and some basic theories of stability associated with linear systems, while the nonlinear part is done by Lyapunov theory.  The contributions of this paper are summarized as follows:

(1) Both continuous- and discrete-time distributed algorithms are provided for solving CFPs. Different from the distributed algorithms for solving linear equations in \cite{S. Mou15, H. Cao17}, in which the algorithms need to restrict each agent's initial state within the solution set of its corresponding equations, the CFPs can be solved by the presented algorithms under arbitrary initial states.

(2) The continuous-time distributed gradient-based algorithm has also been investigated in \cite{Touri B34}, where convergence of the algorithm relies on a time-varying parameter. Our algorithm does not involve a time-varying parameter and it does not require the assumption on boundedness of the subgradient as in \cite{Touri B34}. We prove that, if the directed graph is fixed and strongly connected, all agents' states will reach a common point asymptotically and the point is located in the solution set of the CFP. Moreover, we find that the CFP can be solved if the $\delta-$graph associated with a time-varying graph is strongly connected.

(3) Discrete-time distributed subgradient-based algorithms have been studied in \cite{A. Nedi24}, where the communication graph is balanced. Unlike \cite{A. Nedi24, S.S. Ram 07}, in our algorithm, only relative information between the agents is required and the convergence can also be ensured when the communication graph is unbalanced. We prove that the effectiveness of the presented algorithm can be guaranteed when the directed graph is strongly connected.

This paper is organized as follows. In Section \ref{se2}, we present some notions in graph theory and state the problem studied in this paper. In Section \ref{se3}, centralized algorithms in both continuous- and discrete-time cases for the CFP are focused on and the convergence of them is analyzed. In Section \ref{se4},  the distributed control algorithm in continuous-time case is presented for the MAS to solve the CFP, and the convergence is analyzed under both fixed and time-varying communication graphs. The discrete-time case is studied in Section \ref{se5}. In Section \ref{se6}, a distributed gradient-based algorithm is designed for a CFP involving linear inequalities. Simulation examples are presented in Section \ref{se7}. Section \ref{se8} concludes the whole paper.

{\bf Notation:} Throughout this paper, we use $|a|$ to represent the absolute value of scalar $a$. $\mathbb{R}$ and $\mathbb{C}$ denote the set of real number and the set of complex number,respectively. Let $\mathbb{R}^{m}$ be the $m$-dimensional real vector space and $\mathbb{C}^{m}$ be the complex one. For a given vector $x\in\mathbb{R}^{m}$, $x>0(\geq 0)$ implies that each entry of vector $x$ is greater than (not less than) zero. $\|x\|$ denotes the standard Euclidean norm, i.e., $\|x\|=\sqrt{x^Tx}$. For a function $g(\cdot):\mathbb{R}^{m}\rightarrow\mathbb{R}$, we denote its plus function by $g^{+}(\cdot)=\max [g(\cdot), 0]$. $\textbf{1}_n$ denotes the $n$-dimensional vector with elements being all ones. $I_n$ denotes the $n\times n$ identity matrix. The transposes of matrix $A$ and vector $x$ are denoted as $A^T$ and $x^T$, respectively. For any two vectors $u$ and $v$, the operator $\langle u, v\rangle$ denotes the inner product of $u$ and $v$. For matrices $A$ and $B$, the Kronecker product is denoted by $A\otimes B$.

\section{Preliminary and problem formulation}\label{se2}

\subsection{Graph theory}
The communication topology is denoted by $\mathcal{G}(\mathcal{A}(t))=(\mathcal{V},\mathcal{E}(t),\mathcal{A}(t))$, $\mathcal{V}$ is a set of vertices, $\mathcal{E}(t)\subset\mathcal{V}\times\mathcal{V}$ is an edge set, and the weighted matrix $\mathcal{A}(t)=(a_{ij}(t))_{n\times n}$ is a non-negative matrix for adjacency weights of edges. If node $i$ can receive the information from node $j$, then node $j$ is called
 as node $i$'s neighbor and it is denoted by $(j,i)\in \mathcal{E}(t)$ and $a_{ij}(t)>0$. Otherwise, $a_{ij}(t)=0$. Denote $N_{i}(t)=\{j\in\mathcal{V}|(j,i)\in\mathcal{E}(t)\}$ to represent the neighbor set of node $i$ at time $t$. The Laplacian
matrix of the graph is defined as $L(t)=(l_{ij}(t))_{n\times n}$, where $l_{ij}(t)=-a_{ij}(t)$ if $i\neq j$ and $l_{ij}(t)=\sum\limits_{j = 1}^na_{ij}(t)$ if $i=j$ for any $i=1, \cdots, n$. For a fixed and directed graph $\mathcal{G}(\mathcal{A})$, a path of length $r$ from node $i_1$ to node $i_{r+1}$ is a sequence of $r + 1$ distinct nodes $i_1 \cdots, i_{r+1}$ such that $(i_q, i_{q+1}) \in\mathcal{E}$ for $q=1,\cdots, r$. If there exists a path between any two nodes in $\mathcal{V}$, then $\mathcal{G}(\mathcal{A})$ is said to be strongly connected. A directed graph, where every node has exactly one neighbor except the root, is said to be a directed tree. A spanning tree of a directed
graph is a directed tree formed by graph edges that connect all the
nodes of the graph \cite{C. Godsil 06}. We say that a graph has a spanning tree if a subset of the edges forms a spanning tree.

For a time-varying and directed graph $\mathcal{G}(\mathcal{A}(t))$, $(j,i)$ is called a $\delta-$edge if there always exist two positive constants $T$ and $\delta$ such that $\int_{t}^{t+T}a_{ij}(s)ds\geq\delta$ for any $t\geq0$. A $\delta-$graph, induced by $\mathcal{G}(\mathcal{A}(t))$, is defined as $\mathcal{G}_{(\delta, T)}=(\mathcal{V},\mathcal{E}_{(\delta, T)})$, where $\mathcal{E}_{(\delta, T)}=\big\{(j,i)\in\mathcal{V}\times\mathcal{V}|\int_{t}^{t+T}a_{ij}(s)ds\geq\delta~for~any~$ $ t\geq0\big\}$.
 The communication graph $\mathcal{G}(\mathcal{A}(t))$ is said to be balanced if the sum of the interaction weights from and to an agent $i$ are equal, i.e., $ \sum\limits_{j=1}^{n}a_{ij}(t)=\sum\limits_{j=1}^{n}a_{ji}(t)$.
 \begin{lemma}\label{le5}\cite{W. Ren6}
For a fixed graph $\mathcal{G}(\mathcal{A})$, if $\mathcal{G}(\mathcal{A})$ has a spanning tree, then the Laplacian matrix $L$ has one simple 0 eigenvalue and the other eigenvalues have positive real parts.
\end{lemma}
\begin{lemma}\cite{Olfati-Saber5}\label{le9}
For a fixed graph $\mathcal{G}(\mathcal{A})$, if $\mathcal{G}(\mathcal{A})$ is strongly connected, then there exists a vector $w = \left[ w_1  \cdots w_n \right]^T>0$ such that $w^TL=0$.
\end{lemma}

For ease of description, if $\mathcal{G}(\mathcal{A})$ has a spanning tree, we use $\lambda_1(L)$ to represent the 0 eigenvalue and $\lambda_i(L), i=2,\cdots, n$ to represent other non-zero eigenvalues.

\subsection{Convex analysis}
A function $f(\cdot):\mathbb{R}^m\rightarrow\mathbb{R}$ is convex if it holds $f(\gamma x + (1 - \gamma )y) \le \gamma f(x) + (1 - \gamma )f(y)$ for any $x\neq y\in\mathbb{R}^m $ and $0<\gamma <1$. For convex function $f(x)$, if $\langle {\begin{array}{*{20}c}
   {\nabla f(x),} & {y - x}  \\
\end{array}} \rangle \le f(y) - f(x)$ holds for any $y\in\mathbb{R}^m$, then $\nabla f(x)$ is a subgradient of function $f$ at point $x\in\mathbb{R}^m$. There must exist subgradients for any convex function. Furthermore, if the convex function is differentiable, its gradient is the unique subgradient.

Given a set $\Omega\subset\mathbb{R}^m$, it is called as a convex set if $ \gamma x + (1 - \gamma )y \in \Omega$ for any scalar $0<\gamma <1$ and $x, y\in\Omega$. For a closed convex set $\Omega$, let $
\left\| x \right\|_\Omega  \mathop  = \limits^\Delta  \inf _{y \in \Omega } \left\| {x - y} \right\|$ denote the standard Euclidean distance
of vector $x\in\mathbb{R}^m$ from $\Omega$. Then, there is a unique element $P_\Omega  (x)\in \Omega$ such that $\left\| {x - P_\Omega  (x)} \right\| = \left\| x \right\|_\Omega$, where $P_\Omega  (\cdot)$ is called the projection onto the set $\Omega$ \cite{J. Aubin 01}. Moreover, $P_\Omega  (\cdot)$ has the non-expansiveness property: $\|P_\Omega  (x)-P_\Omega  (y)\|\leq\|x-y\|$ for any $x,y\in\mathbb{R}^m$.
\begin{lemma}\label{LE}
For a convex function $g(\cdot):\mathbb{R}^{m}\rightarrow\mathbb{R}$, suppose the set $X=\{x\in\mathbb{R}^{m}|g^{+}(x)=0\}$ is non-empty, it holds $z\in X$ if and only if 0 is a subgradient of the plus function $g^{+}$ at point $z$.
\end{lemma}
\begin{proof}
\emph{Sufficiency}. By the definition of $g^+(\cdot)$, we know function $g^+(\cdot)$ is convex. Therefore, the subgradient of function $g^{+}(\cdot)$ always exists. If 0 is a subgradient of the plus function $g^{+}$ at point $z$, by the definition of the subgradient, we have $g^{+}(y)-g^{+}(z)\geq0^T(y-z)=0$ for any $y\in\mathbb{R}^{m}$. Let $y\in X$, then we have $-g^{+}(z)\geq0$. By this and the fact that $g^{+}(z)\geq0$, it can be concluded that $g^{+}(z)=0$.

\emph{Necessity}. If $z\in X$, we have $g^{+}(z)=0$. Due to the fact that $g^{+}(y)\geq0$, we have $g^{+}(y)-0\geq0^T(y-z)$ for any $y\in\mathbb{R}^{m}$. Thus, 0 is a subgradient of the plus function $g^{+}$ at point $z$.
\end{proof}
\begin{lemma}\label{le1}\cite{A. Nedi24}
Given a closed convex set $\Omega\subset\mathbb{R}^m$, it holds
\[
\left\langle {P_\Omega  (x) - x, x - y} \right\rangle \leq -\|x\|_{\Omega}^2
\] for any $x\in \mathbb{R}^m, y\in\Omega$.
\end{lemma}

\subsection{Problem formulation}
Consider a MAS consisting of $n$ agents, labeled by set $\mathcal{V}=\{1, \cdots, n\}$. Here we consider agents with both continuous-time dynamics
\begin{equation}\label{eq1}
\dot{x}_i (t) =u_i(t), i\in\mathcal{V}
\end{equation}
and discrete-time dynamics
\begin{equation}\label{eq18}
{x}_i (t+1) ={x}_i (t)+u_i(t), i\in\mathcal{V}
\end{equation}
where $x_i (t) \in \mathbb{R}^m$ and $u_i (t) \in\mathbb{R}^m$ are respectively, the state and input of agent $i$. The objectives of this paper are to design $u_i (t)$ for (\ref{eq1}) and (\ref{eq18}) by using only local information to solve the following CFP:
\begin{equation}\label{eq2}
\left\{ {\begin{array}{*{20}c}\begin{split}
   &{g_i (x) \le 0}  \\
   &{x \in X:\mathop  = \limits^\Delta   \cap _{i = 1}^n X_i }  \\
\end{split}\end{array}} \right.\begin{array}{*{20}c}
   {} & {}  \\
\end{array}i = 1, \cdots ,n
\end{equation}
where $ x\in \mathbb{R}^{m}$, $g_{i}(\cdot):\mathbb{R}^{m}\rightarrow\mathbb{R}$ is a convex function, it is continuous on $(-\infty, \infty)$. Each $X_i$ is a closed convex set. Agent $i$ can only have access to the information associated with subgradient $\nabla g_{i}^{+}(\cdot)$ and projection $P_{X_i}(\cdot)$. We assume each $\nabla g_{i}^{+}(\cdot)$ is piecewise continuous for any $i = 1, \cdots ,n$.

\begin{remark} Note that if and only if $x\in X_i$, it holds $x=P_{X_i}(x)$. If $x=P_{X_i}(x)$ for all $i=1, \cdots ,n$, then $x$ belongs to their intersection. Since the algorithms in the following sections refer to the projection operator $P_{X_i}(\cdot)$, here we only consider some convex sets $X_i$ onto which the projection $P_{X_i}(x)$ can be easily calculated or their expressions could be given in detail at any point $x$. For example, if set $X$ represents the solution set of linear equation $a^Tx-b=0$, i.e., $X=\{x|a^Tx-b=0\}$, where $a, x\in \mathbb{R}^{m}, b\in \mathbb{R}$, it is easy to show that $P_{X}(x)=\left(I-\frac{aa^T}{\|a\|^2}\right)x+\frac{ba}{\|a\|^2}$ is a projection of $x$ onto set $X$. Consequently, it is not difficult to find that the algorithms in the following sections are also available to the CFP involving linear equations.
 \end{remark}
The solution set of CFP (\ref{eq2}) is denoted by $\textbf{X}^*$ and the following assumption is adopted throughout the paper.
\begin{assumption}\label{A1}
$\textbf{X}^*$ is non-empty.
\end{assumption}
Note that a vector $x^{*}$ belongs to $\textbf{X}^*$, if and only if it holds that $x^{*}\in X$  and $g_{i}^{+}(x^{*})=0$ for each $i\in\{1,\cdots,n\}$.

\section{Centralized algorithms for CFPs}\label{se3}
In this section, we focus on the following CFP
\begin{equation}\label{eq3}
\left\{ {\begin{array}{*{20}c}\begin{split}
   &{g (x) \le 0}  \\
   &{x \in X }  \\
\end{split}\end{array}} \right.\begin{array}{*{20}c}
   {} & {}  \\
\end{array}
\end{equation}
where $ x\in \mathbb{R}^{m}$, $g(\cdot):\mathbb{R}^{m}\rightarrow\mathbb{R}$ is a convex function, and $X$ is a closed convex set.

\subsection{Continuous-time case}
To solve CFP (\ref{eq3}), the following continuous-time subgradient and projection-based algorithm is proposed.
\begin{equation}\label{eq4}
\begin{split}
\dot x(t) =  - \alpha (t)[x(t) - P_X (x(t))]- \beta (t)\nabla g^ +  (x(t))
\end{split}
\end{equation}
where $\alpha (t)$,$\beta (t)\in \mathbb{R}$.
\begin{theorem}
Suppose CFP (\ref {eq3}) has a non-empty solution set $\textbf{X}^*$, if $\alpha (t)\geq0$ and $\beta (t)\geq0$ satisfy that $\int_0^\infty  {\alpha (t)}  \to \infty$ and $\int_0^\infty  {\beta(t)} \to \infty$, then $x(t)$ in (\ref{eq4}) converges to a vector $x^*$ in set $\textbf{X}^*$.
\end{theorem}
\begin{proof}
Define a positive-definite Lyapunov function candidate $V(t)=\frac{1}{2}\|x(t)-x_0\|^2$, where $x_0\in\textbf{X}^*$. By the definition of $g^+$, it holds $g^+(x_0)=\|x_0\|_X=0$. Based on the property of the subgradient, we have $\left\langle x(t) - x_0, \nabla g^ +  (x(t)) \right\rangle \geq g^+  (x(t))$. Taking
the derivative of function $V(t)$ with respect to $t$ yields
\begin{equation}\label{eq5}
\begin{split}
\dot V(t) &= \left\langle x(t) - x_0, \dot x(t)\right\rangle\\
&=\left\langle x(t) - x_0, - \alpha (t)[x(t) - P_X (x(t))] - \beta (t)\nabla g^ +  (x(t))\right\rangle\\
&=- \alpha (t)\left\langle x(t) - x_0, x(t) - P_X (x(t))\right\rangle- \beta (t)\langle x(t) - x_0, \nabla g^ + (x(t))\rangle\\
&\leq- \alpha (t)\left\langle x(t) - x_0, x(t) - P_X (x(t))\right\rangle- \beta (t)g^ +  (x(t)).
\end{split}
\end{equation}
By Lemma \ref{le1}, we know $- \left\langle x(t) - x_0, x(t) - P_X (x(t))\right\rangle\leq-\|x(t)\|_X^2\leq0$. Note that $g^ +  (x(t))\geq 0$. Thus, $\dot V(t)\leq0$. Moreover, $V(t)$ is bounded by zero, it can be concluded that $V(t)$ converges and $V(\infty)$ exists, which implies $\|x(t)-x_0\|$ converges. By inequality (\ref{eq5}), we have
\[
\int_0^\infty  {\alpha (t)\left\| {x(t)} \right\|} _X^2 d_t  + \int_0^\infty  {\beta (t)g^ +  } (x(t))d_t  \le V(0) - V(\infty ) < \infty.
\]
Since ${\alpha (t)\left\| {x(t)} \right\|} _X^2$ and ${\beta (t)g^ +  } (x(t))$ are both non-negative, then we have $\int_0^\infty  {\alpha (t)\left\| {x(t)} \right\|} _X^2$ $ d_t< \infty$ and $\int_0^\infty  {\beta (t)g^ +  } (x(t))d_t < \infty$. These and the facts $\int_0^\infty  {\alpha (t)}  \to \infty$ and $\int_0^\infty  {\beta(t)} \to \infty$ imply $\lim\limits_{t\rightarrow\infty}\inf \left\| {x(t)}-P_X(x(t)) \right\|= \lim\limits_{t\rightarrow\infty}\inf g^ +(x(t))=0$. Thus, there exists a subsequence $\{x(t_k)\}$ of ${x(t)}$ such that $\lim\limits_{k\rightarrow\infty}x(t_{k})=\lim\limits_{t\rightarrow\infty}\inf x(t)=x^*$, where $x^*$ is a point in the solution set of CFP (\ref{eq3}).  Moreover, note that $V(x(t))$ converges, it can be concluded that
$\lim\limits_{t\rightarrow\infty}x(t)=x^*\in\textbf{X}$. Hence, the validity of the result is verified.
\end{proof}

\begin{corollary}
Suppose CFP (\ref {eq3}) has a non-empty solution set $\textbf{X}^*$, if $x(t)$ adjusts its value with the following dynamics
 \[
\begin{split}
\dot x(t) =  - [x(t) - P_X (x(t))]- \nabla g^ +  (x(t))
\end{split}
\]
then $x(t)$ converges to a vector $x^*$ in set $\textbf{X}^*$.
\end{corollary}
\subsection{Discrete-time case}
Now we present the discrete-time algorithm for CFP (\ref{eq3}).
\begin{equation}\label{eq19}
\begin{split}
\left\{ {\begin{array}{*{20}c}\begin{split}
   &{\xi (t) = x(t) - \beta (t) {\nabla g^ +  \left( {x(t)} \right)} } \\
   &{\varphi (t) = \alpha (t)\left( {\xi (t) - P_X (\xi (t))} \right)}  \\
   &{x(t + 1) = \xi (t) - \varphi (t)}  \\
\end{split}\end{array}} \right.
\end{split}
\end{equation}
where $P_X (\cdot)$ and $\nabla g^ +  (x(t))$ are defined as those in (\ref{eq4}).
\begin{assumption}\label{A5}
$\nabla g^ +  (x(t))\leq K$ for some $K\geq 0$.
\end{assumption}
\begin{lemma}\label{le8}\cite{D. Richert31}
Let $\{z(t)\}$ be a non-negative scalar sequence such that
\[
z(t+1) \le (1 + a(t))z(t) - b(t) + c(t)
\]
for all $t\geq0$, if $a(t)\geq0, b(t)\geq0, c(t)\geq0$ with $\sum\limits_{t = 0}^\infty  {a(t)}<\infty$ and $\sum\limits_{t = 0}^\infty  {c(t)}<\infty$, then the sequence
$\{z(t)\}$ converges to some constant $z^*$ and $\sum\limits_{t = 0}^\infty {b(t)}<\infty$.

\end{lemma}

\begin{theorem}
Under Assumptions \ref{A5}, if CFP (\ref {eq3}) has a non-empty solution set $\textbf{X}^*$, and $\alpha (t)$ , $\beta (t)$ satisfy

(a)  $\alpha(t) \in [0, 1]$, $\sum\limits_{t = 0}^\infty {\alpha (t)}  \to \infty$ and $\sum\limits_{t = 0}^\infty {\alpha^2 (t)}  < \infty$;

(b)  $0\leq\beta(t)\leq\infty$, $\sum\limits_{t = 0}^\infty {\beta(t)} \to \infty$ and $\sum\limits_{t = 0}^\infty {\beta^2 (t)} < \infty$.

 Then, $x(t)$ in (\ref{eq19}) converges to a vector $x^*$ in set $\textbf{X}^*$.
\end{theorem}
\begin{proof}
We choose the Lyapunov function candidate as $V(t)=\|x(t)-x_0\|^2$, where $x_0\in\textbf{X}^*$. Taking the difference of function $V (t)$ along with (\ref{eq19}) yields
\begin{equation}\label{eq20}
\begin{split}
 \Delta V(t)&=V(t+1)-V(t)\\
& =\|\xi (t) -\varphi (t)-x_0\|^2-\|x(t)-x_0\|^2\\
&=\|(1-\alpha(t))(\xi (t)-x_0)+\alpha(t)(P_X (\xi (t))-x_0)\|^2-\|x(t)-x_0\|^2\\
&\leq\Big((1-\alpha(t))\|\xi (t)-x_0\|+\alpha(t)\|P_X (\xi (t))-x_0\|\Big)^2-\|x(t)-x_0\|^2\\
&\leq\|\xi (t)-x_0\|^2-\|x(t)-x_0\|^2
\end{split}
\end{equation}
where the last inequality follows from the non-expansiveness property of projection operator, i.e., $\|P_X (\xi (t))-x_0)\|\leq \|\xi (t)-x_0\|$. Moreover, we have
\begin{equation}\label{eq21}\begin{split}
\left\| {\xi (t) - x_0 } \right\|^2 &\le \left\| {x(t) - x_0 } \right\|^2  - 2\beta (t)\langle\nabla g^ +  (x(t)), x(t) - x_0 \rangle  \\
&~~~+ \beta ^2 (t)K \\
&\le \left\| {x(t) - x_0 } \right\|^2  - 2\beta (t){\left( {g^ +  (x(t)) - g^ +  (x_0 )} \right)}\\
&~~~+ \beta ^2 (t)K.\\
 \end{split}
\end{equation}
From inequalities $(\ref{eq20})$ and $(\ref{eq21})$, we have $\Delta V(t)\leq \beta ^2 (t)K $. Thus, it holds that $V(t)\leq V(0)+\sum\limits_{t = 0}^{t-1} {\beta ^2 (t) {K }}$  $\leq V(0)+\sum\limits_{t = 0}^\infty  {\beta ^2 (t)K } <\infty$. By the definition of $V(t)$, it can be concluded that $x(t)$ is bounded. Since $\|\beta (t) \nabla g^ +  \left( {x(t)} \right)\|<\infty$, $\xi(t)$ is bounded. This and the continuity of $P_X(\xi (t))$ imply $\|{\xi (t) - P_X (\xi (t))}\|<\infty$. Denote $\nabla (t)=\beta (t){\nabla g^ +  \left( {x(t)} \right)}$, since  $\sum\limits_{t = 0}^\infty {\alpha^2 (t)} < \infty$ and $\sum\limits_{t = 0}^\infty {\beta^2 (t)} < \infty$, it can be concluded that $\sum\limits_{t = 0}^\infty\|\nabla (t)\|^2<\infty$ and $\sum\limits_{t = 0}^\infty\|\varphi (t)\|^2<\infty$. Similar to (\ref{eq20}), we also have
\begin{equation}\label{eq22}
\begin{split}
 \Delta V(t)&=V(t+1)-V(t)\\
&=-2\langle \nabla (t)+\varphi (t), x(t) - x_0\rangle +\|\nabla (t)+\varphi (t)\|^2 \\
&=-2\langle\nabla (t), x(t) - x_0\rangle -2\langle \varphi (t), \xi(t) - x_0  \rangle\\
&~~~-2\langle\varphi (t), \nabla (t)\rangle+\|\nabla (t)+\varphi (t)\|^2\\
&=-2\langle \nabla (t), x(t) - x_0\rangle -2\langle\varphi (t), \xi(t) - x_0\rangle+\|\nabla (t)\|^2+\|\varphi (t)\|^2\\
&\leq - 2\beta (t)g^ +  (x(t)) -2\alpha(t) \|\xi(t)\|^2_X+\|\nabla (t)\|^2+\|\varphi (t)\|^2\\
&=- 2\beta (t) g^ +  (x(t)) - 2\alpha(t)\|\xi(t)\|^2_X+\|\nabla (t)\|^2+\|\varphi (t)\|^2.
\end{split}
\end{equation}
Recall the fact that $\sum\limits_{t = 0}^\infty \|\nabla (t)\|^2+\|\varphi (t)\|^2<\infty$ and $- 2\beta (t) g^ +  (x(t))-2\alpha(t)\|\xi(t)\|^2_X<0$, by Lemma \ref{le8}, it can be concluded $\|x(t)-x_0\|$ converges and it holds $\sum\limits_{t = 0}^\infty \Big(\beta (t) g^ +  (x(t))+ \alpha(t)\|\xi(t)\|^2_X\Big)<\infty$. Since $ \beta (t) g^ +  (x(t))>0$ and $\alpha(t)\|\xi(t)\|^2_X>0$ for all $t>0$, we have $\sum\limits_{t = 0}^\infty \beta (t) g^ +  (x(t))<\infty$ and $\sum\limits_{t = 0}^\infty\alpha(t)\|\xi(t)\|^2_X<\infty$. By the facts $\sum\limits_{t = 0}^\infty {\alpha (t)}\rightarrow\infty$ and $\sum\limits_{t = 0}^\infty {\beta (t)}\rightarrow\infty$, we have $\lim\limits_{t\rightarrow\infty}\inf \left\| {\xi(t)}-P_X(\xi(t)) \right\|= \lim\limits_{t\rightarrow\infty}\inf g^ +(x(t))=0$. Thus, there exists a subsequence $\{x(t_k)\}$ of ${x(t)}$ such that $\lim\limits_{k\rightarrow\infty}x(t_{k})=x^*$, where $x^*$ is a vector such that $g^+(x^*)=0$. By the fact $\|x(t)-x_0\|$ converges, we can conclude $\lim\limits_{t\rightarrow\infty} x(t)=x^*$. Furthermore, note that $\nabla (t)\rightarrow 0$ as $t\rightarrow\infty$, thus $\lim\limits_{t\rightarrow\infty}\inf \left\| {\xi(t)}-P_X(\xi(t)) \right\|=0$ and $\lim\limits_{t\rightarrow\infty} x(t)=x^*$ imply $\lim\limits_{t\rightarrow\infty}\left\| {x^*}-P_X(x^*) \right\|=0$.  Therefore, $x^*$ is a solution to CFP (\ref{eq3}), i.e., $x^*\in \textbf{X}^*$.
\end{proof}

\section{Continuous-time distributed control algorithms for solving CFPs}\label{se4}

In this section, we focus on solving CFP (\ref{eq2}) for continuous-time MAS (\ref{eq1}) in a distributed  manner, which means that each agent has access to only its own state and that from its neighbors. The following input is proposed.
\begin{equation}\label{eq6}
\left\{ {\begin{array}{*{20}c}\begin{split}
   &{u_i(t) = \sum\limits_{i \in N_i (t)} {a_{ij} (t)} (x_j (t) - x_i (t)) + \phi _i (t)}  \\
   &{\phi _i (t) = -\tau\big([x_i(t) - P_{X_i} (x_i(t))]+ \nabla g_i^ +  (x_i(t))\big)} \\
\end{split}\end{array}} \right.
~~~~i\in\mathcal{V}
\end{equation}
where $\tau$ is a positive coefficient. Note that $\phi _i$ depends on only agent $i$'s own state, so (\ref{eq6}) is distributed. Based on Lemma \ref{LE} in Section II, here we set $\nabla g_i^+(x) = 0$ if $g_i(x)\leq0$ and $\nabla g^+_i(x) = \nabla g_i(x)$ otherwise.

\begin{remark}
If we set $\tau=0$ in algorithm (\ref{eq6}), then it will become a typical linear consensus algorithm for MASs studied in  \cite{W. Ren6}, \cite{Olfati-Saber5}. In this case, MASs reach consensus asymptotically if the communication graph is fixed and has a spanning tree. The distributed subgradient-based algorithm was studied for continuous-time multi-agent systems to optimize a sum of convex objective functions in \cite{Touri B34}, but the convergence of the algorithm relies on a time-varying parameter and the projection term was not involved.
\end{remark}

Let $x(t)=\big[ x^T_1(t), \cdots, x^T_n(t) \big]^T$ and $\phi(t)=\big[ \phi^T_1(t), \cdots, \phi^T_n(t) \big]^T$, MAS (\ref{eq1}) with (\ref{eq6}) can be rewritten as
\begin{equation}\label{eq7}
\dot{x}(t)=-\left(L(t) \otimes I_m\right) x(t)+\phi(t).
\end{equation}
 \begin{lemma}\label{le3}\cite{Kaihong Lu33}
Let $b(t)$ be a bounded function, if $\lim\limits_{t\rightarrow\infty}b(t)=b$ and $0<\gamma<1$, then $\lim\limits_{t\rightarrow\infty} \int_{0}^t \gamma^{t-s}b(s)$ $ ds=-\frac{b}{ln\gamma}$.
\end{lemma}
\begin{lemma}\label{le4}\cite{B. Touri32}
Given a symmetric matrix $P=(p_{ij})_{n\times n}$ with 0 eigenvalue and a vector $x = [x_1, \cdots, x_n]^T$, if $P\textbf{1}_n=0$, then it holds $x^T Px =  - \sum\limits_{i = 1}^n {\sum\limits_{j = i + 1}^n p_{ij}({ x_i -x_j } )^2}$.
\end{lemma}
\begin{lemma}\label{le6}
Given a linear system $\dot x(t) = Ax(t) + u(t)$, if the state matrix $A\in\mathbb{R}^{n\times n}$ is Hurwitz stable and $u(t)\in\mathbb{R}^{n} $ satisfies $\|u(t)\|<\infty$ and $\lim\limits_{t\rightarrow\infty}u(t)=0$, then the linear system is asymptotically stable to zero, i.e., $\lim\limits_{t\rightarrow\infty}{x}(t)=0$.
\end{lemma}
\begin{proof} Since matrix $A$ is Hurwitz stable, all of its eigenvalues have negative real parts. Based on theory of Schur's unitary triangularization, there exists a unitary matrix $U\in\mathbb{C}^{n\times n}$ such that
\[
U^H AU = \left[ {\begin{array}{*{20}c}
   {\lambda _1 } & \lambda _{12} &  \cdots  & \lambda _{1n}  \\
   0 & {\lambda _2 } & \lambda _{23} & \lambda _{2n}  \\
   \vdots & \vdots &  \ddots  & \vdots  \\
   0 & 0 & 0 & {\lambda _n }  \\
\end{array}} \right]\mathop  = \limits^\Delta\Lambda
\]
where $\lambda _i$ is the eigenvalue of matrix $A$, $i=1,\cdots, n$; $U^H$ is the conjugate transpose matrix of $U$. Denote $y(t)=U^Hx(t)$ and $r(t)=U^Hu(t)$, we have $\dot{y}(t)=\Lambda y(t)+r(t)$. By the fact that $\lim\limits_{t\rightarrow\infty}u(t)=0$, we have $\lim\limits_{t\rightarrow\infty}r(t)=0$. Let $y(t) = [ y_1(t), \cdots, y_n (t)]^T$ and  $r(t) = [r_1(t),\cdots, r_n (t)]^T$, we have
$\dot y_n (t) = \lambda _n y_n (t) + r_n (t)$. The term $r_n(t)$ can be viewed as an control input of the linear system and we have $ {y_n (t)} = e^{\lambda _n t}  {y_n (0)}  + \int_0^t {e^{\lambda _n (t - \tau )} } {r_n (\tau )} d_\tau$. Since the real part of $\lambda _n$ is negative, it holds $0<e^{\lambda _n}<1$. By Lemma \ref{le3}, it can be concluded that $\lim\limits_{t\rightarrow\infty}y_n(t)=0$. Since $\dot y_i (t) = \lambda _i y_i(t) +
\Big(\sum\limits_{j = 1}^n {\lambda _{i(i + j)} y_{i + j} (t)} $ $+r_i (t)\Big)$. Through the similar approach for $y_n(t)$, we can conclude $\lim\limits_{t\rightarrow\infty}\left(\sum\limits_{j = 1}^n {\lambda _{i(i + j)} y_{i + j} (t)}+r_i (t)\right)$ $=0$. Reusing Lemma \ref{le3} yields $\lim\limits_{t\rightarrow\infty}y_i(t)=0$ for any $i=1, \cdots, n$. This and the fact $x(t)=Uy(t)$ imply $\lim\limits_{t\rightarrow\infty}x(t)=0$.
\end{proof}

To prove the fact that MAS (\ref{eq1}) with (\ref{eq6}) solves CFP (\ref{eq2}), it is necessary to analyze the convergence of MAS (\ref{eq1}) with (\ref{eq6}). Obviously, the conditions for convergence depend on the connectivity of the graphs. In the following, we will provide the convergence conditions under the fixed graph and the time-varying graph, respectively.

\subsection{Convergence under the fixed communication graph}
\begin{proposition}
Suppose $\|\phi _i(t)\|<\infty$ and $\lim\limits_{t\rightarrow\infty}\phi _i(t)=0$ in $(\ref{eq6})$, $i\in\mathcal{V}$, if the fixed graph $\mathcal{G}(
\mathcal{A})$ is directed and has a spanning tree, then MAS (\ref{eq1}) with (\ref{eq6}) reaches consensus asymptotically.
\end{proposition}
\begin{proof}
Define a variable $\hat {x}(t) = \sum\limits_{i = 1}^n {\frac{{w_i x_i(t) }}{{\sum\limits_{i = 1}^n {w_i } }}}=\left( {\frac{w^T}{\textbf{1}^T w} \otimes I_m } \right)x(t)$, where $w = \left[ w_1  \cdots w_n \right]^T$ is $L$'s left eigenvector associated with 0 eigenvalue.
Based on (\ref {eq7}), we have $\dot{\hat{x}}(t)=\frac{{\left( {w^T  \otimes I_m } \right)}}{{\textbf{1}^T w}}u(t)$. Denote $e_i(t)=x_i(t)-\hat {x}(t)$ and  $e(t)=\big[ e^T_1(t), \cdots, e^T_n(t) \big]^T$. Note that if $\lim\limits_{t\rightarrow\infty}e(t)=0$, then MAS (\ref{eq1}) with (\ref{eq6}) reaches consensus. From (\ref{eq7}), we have
\begin{equation}\label{eq8}\begin{split}
\dot e(t) &=  - (L \otimes I_m )x(t) + \left( {\left( {I_n  - \frac{{\textbf{1}_n w^T }}{{\textbf{1}_n ^T w}}} \right) \otimes I_m } \right)\phi(t)\\
&=  - (L \otimes I_m )x(t)+  (L \otimes I_m )\left( \frac{{\textbf{1}_n w^T }}{{\textbf{1}_n ^T w}}\otimes I_m\right)x(t)\\
 &~~~~+ \left( {\left( {I_n  - \frac{{\textbf{1}_n w^T }}{{\textbf{1}_n ^T w}}} \right) \otimes I_m } \right)\phi(t)\\
&=- (L \otimes I_m )e(t) + \left( {\left( {I_n  - \frac{{\textbf{1}_n w^T }}{{\textbf{1}_n ^T w}}} \right) \otimes I_m } \right)\phi(t)
\end{split}
\end{equation}
where the second equation holds for the fact that $L\textbf{1}_n=0$. Note that ${\frac{1}{\sqrt{w^T w}} }L^Tw =0$. Now we use ${\frac{1}{\sqrt{w^T w}} }w$ to form a set of orthonormal basis on $\in\mathbb{C}^{n}$, denoted by ${\frac{1}{\sqrt{w^T w}} }w, p_2, \cdots,  p_n$. We define $P=({\frac{1}{\sqrt{w^T w}} }w, p_2, \cdots,   p_n)$. It is obvious that $P$ is a unitary matrix, so we can denote

\[
P^T LP = \left[ {\begin{array}{*{20}c}
   {\underline {\begin{array}{*{20}c}
   {\left. ~~~~0 ~~\right|} & 0 &  \cdots  & 0  \\
\end{array}} }  \\
{\begin{array}{*{20}c}
{\left. {\begin{array}{*{20}c}
*\\
\vdots\\
*\\
\end{array}} \right|} & {} & {L_1 } & {}  \\
\end{array}}  \\
\end{array}} \right].
\]
Since $\mathcal{G}(\mathcal{A})$ has a spanning tree, by Lemma \ref{le5}, $L$ has only one 0 eigenvalue and other eigenvalues have positive real part. This implies $-L_1$ is Hurwitz stable. Now define $\tilde e(t) = (P^T  \otimes I_m )e(t)$. From (\ref{eq8}), we have
\begin{equation}\label{eq9}\begin{split}
\dot {\tilde{e}}(t) &=- (P^TLP \otimes I_m )\tilde{e}(t) + \left( {\left( {P^T  - \frac{{P^T\textbf{1}_n w^T }}{{\textbf{1}_n ^T w}}} \right) \otimes I_m } \right)\phi(t).
\end{split}
\end{equation}

Let $\tilde e(t) =[\tilde e^T_1(t),\tilde e^T_2(t)]^T$, where $\tilde e_1(t)\in \mathbb{R}^{m}$ and $\tilde e_2(t)\in \mathbb{R}^{(n-1)m}$. By (\ref{eq9}), we have
\[\begin{split}
\dot {\tilde{e}}_1(t) &=\left( {\left( {{\frac{1}{\sqrt{w^T w}} }w^T  - \frac{{{\frac{1}{\sqrt{w^T w}} }w^T\textbf{1}_n w^T }}{{\textbf{1}_n ^T w}}} \right) \otimes I_m } \right)\phi(t)=0.
\end{split}
\]
Note that $\tilde e_1(t)={\frac{1}{\sqrt{w^T w}} }(w^T\otimes I_m)e(t)={\frac{1}{\sqrt{w^T w}} }(w^T\otimes I_m)\left( {\left( {I_n  - \frac{{\textbf{1}_n w^T }}{{\textbf{1}_n ^T w}}} \right) \otimes I_m } \right)x(t)=0$. Thus, it holds $\tilde e_1(t)=0$ for any $t\geq 0$. Moreover, we have

\[
\dot {\tilde e} _2  =  - \left( {L_1  \otimes I_m } \right)\tilde e_2  + \left[ {\begin{array}{*{20}c}
   {\left(p_2^T  - \frac{{p_2^T \textbf{1}_n w^T }}{{\textbf{1}_n^T w}}\right){ \otimes I_m }}  \\
    \vdots   \\
   {\left(p_n^T  - \frac{{p_n^T \textbf{1}_n w^T }}{{\textbf{1}_n^T w}}\right){ \otimes I_m }}  \\
\end{array}} \right]\phi(t).
\]
Since $\lim\limits_{t\rightarrow\infty}\phi(t)=0$, by Lemma \ref{le6}, we have $\lim\limits_{t\rightarrow\infty}\tilde e_2(t)=0$. This and the fact that $\lim\limits_{t\rightarrow\infty}\tilde e_1(t)=0$ imply $\lim\limits_{t\rightarrow\infty} e(t)=0$. This leads to the validity of this result.
\end{proof}

\begin{theorem}
If the fixed graph $\mathcal{G}(\mathcal{A})$ is directed and strongly connected, then MAS (\ref{eq1}) with (\ref{eq6}) reaches consensus asymptotically, and the consensus state is located in set $\textbf{X}^*$.
\end{theorem}
\begin{proof}
 Since the graph is strongly connected, by Lemma \ref{le9}, there exists a vector $w = \left[ w_1  \cdots w_n \right]^T>0$ such that $w^TL=0$. Consider a positive-definite Lyapunov function candidate $V(t)=\frac{1}{2}\sum\limits_{i = 1}^nw_i\|x_i(t)-x_0\|^2$, where $x_0\in\textbf{X}^*$. By the definition of $g_i^+$, it holds $g^+(x_0)=\|x_0\|_X=0$. Based on the property of subgradient, we have $\left\langle x_i(t) - x_0, \nabla g^ +_i  (x_i(t)) \right\rangle \geq g^+_i  (x_i(t))$. Taking the derivative of function $V(t)$ with respect to $t$ yields
 \begin{equation}\label{eq10}\begin{split}
\dot V(t) &= \sum\limits_{i = 1}^nw_i\left\langle x_i(t) - x_0, \dot x_i(t)\right\rangle\\
&=\sum\limits_{i = 1}^nw_i\big\langle x_i(t) - x_0, \sum\limits_{i \in N_i (t)} {a_{ij} (t)} (x_j (t) - x_i (t))- \tau[x_i(t)\\
&~~~~ - P_{X_i} (x_i(t))] - \tau\nabla g^+_i  (x_i(t))\big\rangle\\
&=\sum\limits_{i = 1}^n \sum\limits_{i \in N_i (t)}w_i a_{ij}\left\langle x_i(t) - x_0, x_j (t) - x_i (t)\right\rangle\\
&~~~~ -\tau\sum\limits_{i = 1}^n w_i\left\langle x_i(t) - x_0, x(t) - P_{X_i} (x_i(t))\right\rangle\\
&~~~~ -\tau\sum\limits_{i = 1}^n w_i\left\langle x_i(t) - x_0, \nabla g^ +_i (x_i(t)) \right\rangle.\\
\end{split}
\end{equation}
Denote $x(t)=\big[ x^T_1(t), \cdots, x^T_n(t) \big]^T$, we have
\begin{equation}\label{eq11}\begin{split}
\sum\limits_{i = 1}^n \sum\limits_{i \in N_i (t)}w_i a_{ij}\left\langle x_i(t) - x_0, x_j (t) - x_i (t)\right\rangle &=-\left( {x(t) - \left( {\textbf{1}_n  \otimes I_m } \right)x_0 } \right)^T \left( {WL \otimes I_m } \right)x (t)\\
&= -x^T(t) \left( {\frac{{WL + L^T W}}{2} \otimes I_m } \right)x(t)\\
&~~~ + x_0 ^T \left( {w^T L \otimes I_m } \right)x(t) \\
& =   x^T(t) \left( {\frac{{W(-L )+ (-L)^T W}}{2} \otimes I_m } \right)x (t)\\
& =  - \sum\limits_{i = 1}^n {\sum\limits_{j = i + 1}^n {\frac{{w_i a_{ij}  + w_j a_{ji} }}{2}\left\| {x_j(t)  - x_i(t) } \right\|^2 } }  \\
&\le 0 \\
\end{split}\end{equation}
where $W=diag(w)$ is a diagonal matrix formed by $w$ and the last equation results from Lemma \ref{le4}. By Lemma \ref{le1}, we know $-\left\langle x_i(t) - x_0, x_i(t) - P_{X_i} (x_i(t))\right\rangle\leq-\|x_i(t)\|_{X_i}^2\leq0$. Based on (\ref{eq10}) and (\ref{eq11}), we have
\begin{equation}\label{eq12}\begin{split}
\dot V(t)\leq- \tau\sum\limits_{i = 1}^n w_i\left\|x_i(t)\right\|_{X_i}^2-\tau\sum\limits_{i = 1}^n w_ig^ +_i  (x_i(t)).
\end{split}\end{equation}
Note that $g_i^ +  (x_i(t))\geq 0$. Thus, $\dot V(t)\leq0$. Moreover, $V(t)$ is bounded by zero, it can be concluded that $V(t)$ converges and $V(\infty)$ exists, which implies $\|x_i(t)-x_0\|$ converges and $\|x_i(t)\|$ is bounded. By (\ref{eq12}), we have
\[\begin{split}
&\tau\int_0^\infty  {\sum\limits_{i = 1}^n w_i\left\|x_i(t)\right\|_{X_i}^2 d_t}  + \tau\int_0^\infty  {\sum\limits_{i = 1}^n w_ig^ +_i  } (x_i(t))d_t \\
&\le V(0) - V(\infty )\\
 &< \infty.
\end{split}\]
 Thus, it holds $\int_0^\infty  {\left\| {x_i(t)} \right\|} _{X_i}^2$ $ d_t< \infty$ and $\int_0^\infty  {g^ + _i } (x_i(t))d_t < \infty$. These imply $\lim_{t\rightarrow\infty} \big\| {x_i(t)}$ $-P_{X_i}(x_i(t)) \big\| =\lim_{t\rightarrow\infty} g^+_i(x_i(t))=0$ for each $i\in\mathcal{V}$. By the definition of the subgradient $\nabla g_i^+(\cdot)$, we can conclude $\lim\limits_{t\rightarrow\infty}\phi _i (t)=0$ for $i\in\mathcal{V}$. By the continuity of $g_i^+(x_i(t))$ and the boundedness of $\|x_i(t)\|$, it can be concluded $\phi _i (t)$ is bounded. Recall Proposition 1, we know MAS (\ref{eq1}) with (\ref{eq6}) reaches consensus asymptotically, denote $x^*$ as the consensus state, i.e., $\lim\limits_{t\rightarrow\infty}x_i (t)=x^*$ for each $i\in\mathcal{V}$. Therefore, $x^*\in\textbf{X}^*$. The validity of this result is verified.
\end{proof}

\begin{remark}
The strongly connected condition proposed in Theorem 3 is sufficient to solve CFP (\ref{eq2}).
In fact, it is also necessary in many cases. Now we set an example to illustrate that the CFP can not be solved by the MAS if the graph is not strongly connected. Suppose graph $\mathcal{G}$ is not strongly connected, then there exists at least one strongly connected component that can not receive information from others. We denote the set consisting of all agents in this component by $\mathcal{V}_1$. Suppose that all agents in $\mathcal{V}_1$ are constrained by inequality $x\leq 0$. If we set $x_i(0)=0$ for each $i\in\mathcal{V}_1$, then it holds $x_i(t)=0$ for any $t>0$ and $i\in\mathcal{V}_1$. In another strongly connected component, if there exists one agent that is constrained by inequality $x\leq -1$, it is easy to see that the CFP can never be solved under such a graph.
\end{remark}

If communication graph $\mathcal{G}(\mathcal{A})$ is bidirectional and $a_{ij}  = a_{ji}$ for each $i\in\mathcal{V}$, $\mathcal{G}(\mathcal{A})$ becomes an undirected graph. For the undirected case, we state the result as follows.
\begin{corollary}
If the fixed graph $\mathcal{G}(\mathcal{A})$ is undirected and connected, then MAS (\ref{eq1}) with (\ref{eq6}) reaches consensus asymptotically, and the consensus state is in set $\textbf{X}^*$.
\end{corollary}

\subsection{Convergence under the time-varying communication graph}
For system (\ref{eq7}), by the properties of linear systems \cite{Brockett 05}, the solution of system (\ref{eq7}) can be written as follows.
 \begin{equation}\label{eq13}
x(t) = \left( {\Phi (t,s) \otimes I_m } \right)x(s) + \int_s^t {\left( {\Phi (t,\tau) \otimes I_m } \right)u(\tau )} d_\tau
\end{equation}
where ${\Phi(t,s) \otimes I_m }$ is the state-transition matrix from state $x(s)$ to state $x(t)$ with $t\geq s\geq0$. Now, for time-varying graph $\mathcal{G}(t)$, the following assumptions are given.
\begin{assumption}\label{A2}
The communication graph $\mathcal{G}(t)$ is balanced.
\end{assumption}
\begin{assumption}\label{A3}
The $\delta-$digraph $\mathcal{G}_{(\delta, T)}$ is strongly connected.
\end{assumption}
\begin{lemma}\label{le7}\cite{Kaihong Lu33}
 Under Assumptions \ref{A2} and \ref{A3}, for any $t\geq s\geq0$,  $\Phi(t,s) $ in (\ref{eq13}) satisfies the following inequality
\begin{equation}\label{eq14}
\left|\left[\Phi(t,s)\right]_{ij}-\frac{1}{n}\right|\leq\gamma^{t-s},~~~~~i,j\in\{1,\cdots, n\}
\end{equation}
where $\gamma=\left(1-\frac{1}{\left(8n^2\right)^{\lfloor n/2\rfloor}}\right)^{\frac{1}{\left(\lfloor 1/\delta\rfloor+1\right)\left\lfloor n/2\right\rfloor T}}$, the operator $\lfloor x\rfloor$ denotes the largest integer not larger than the value of $x$.
\end{lemma}
\begin{proposition}
 Under Assumptions \ref{A2} and \ref{A3}, if $\|\phi _i(t)\|<\infty$ and $\lim\limits_{t\rightarrow\infty}\phi _i(t)=0$ in $(\ref{eq6})$, $i\in\mathcal{V}$, then MAS (\ref{eq1}) with (\ref{eq6}) reaches consensus asymptotically.
\end{proposition}
\begin{proof}
Since $\mathcal{G}(t)$ is balanced, by Peano-Baker formula
(see \cite{Brockett 05} for detail), it can be concluded that $\Phi (t,s)$ is a double stochastic matrix. Denote $\bar{x}(t)=\frac{1}{n}\sum\limits_{i = 1}^nx_i(t)$, by (\ref{eq13}), we have
\begin{equation}\label{eq15}
\bar x(t) = \frac{1}{n}\left( {\textbf{1}_n^T  \otimes I_m } \right)x(s) + \frac{1}{n}\int_s^t {\left( {\textbf{1}_n^T  \otimes I_m } \right)u(\tau )} d_\tau.
\end{equation}
Based on (\ref{eq13}) and (\ref{eq15}), we have
\begin{equation}\label{eq16}\begin{split}
x(t) - \frac{1}{n}\left(\textbf{1}_n \otimes I_m  \right)\bar x(t) = &\left( {\left( {\Phi (t,0) - \frac{1}{n}\textbf{1}_n \textbf{1}_n^T } \right) \otimes I_m } \right)x(0)\\
&+ \int_s^t {\left( {\left( {\Phi (t, \tau ) - \frac{1}{n}\textbf{1}_n \textbf{1}_n^T } \right) \otimes I_m } \right)u(\tau )} d_\tau.
\end{split}\end{equation}
Applying (\ref{eq14}) in Lemma \ref{le7} to equation (\ref{eq16}) yields
\[
\left\| {x(t) - \frac{1}{n}\left( {\textbf{1}_n^{}  \otimes I_m } \right)\bar x(t)} \right\| \le \sqrt {mn} \gamma^t \left\| {x(0)} \right\| + \sqrt {mn} \int_s^t {\gamma^{t - \tau } \left\| {u(\tau )} \right\|} d_\tau.
\]
Since $0<\gamma=\left(1-\frac{1}{\left(8n^2\right)^{\lfloor n/2\rfloor}}\right)^{\frac{1}{\left(\lfloor 1/\delta\rfloor+1\right)\left\lfloor n/2\right\rfloor T}}<1$ and $\mathop {\lim }\limits_{t \to \infty } \left\| {u(t)} \right\| = 0$,  by Lemma \ref{le3}, we have $\mathop {\lim }\limits_{t \to \infty } \big\| {x(t) - \frac{1}{n}\left( {\textbf{1}_n^{}  \otimes I_m } \right) \bar x(t)} \big\| = 0$. This leads to the validity of this result.
\end{proof}
\begin{theorem}
 Under Assumptions \ref{A1}, \ref{A2} and \ref{A3}, if $\lim\limits_{t\rightarrow\infty}\phi _i(t)=0$ in $(\ref{eq6})$, $i\in\mathcal{V}$, then MAS (\ref{eq1}) with (\ref{eq6}) reaches consensus asymptotically, and the consensus state is in set $\textbf{X}^*$.
\end{theorem}
\begin{proof}
Consider a positive-definite Lyapunov function candidate $V(t)=\frac{1}{2}\sum\limits_{i = 1}^n\|x_i(t)-x_0\|^2$, where $x_0\in\textbf{X}^*$. Taking the derivative of function $V(t)$ with respect to $t$ yields
\begin{equation}\label{eq17}
\begin{split}
\dot V(t) &= \sum\limits_{i = 1}^n\left\langle x_i(t) - x_0, \dot x_i(t)\right\rangle\\
&=\sum\limits_{i = 1}^n\sum\limits_{i \in N_i (t)}{a_{ij} (t)}\big\langle x_i(t) - x_0, x_j (t) - x_i (t)\big\rangle+\sum\limits_{i = 1}^n\big\langle x_i(t) - x_0, \phi _i\big\rangle.\\
\end{split}
\end{equation}
If $\mathcal{G}(t)$ is balanced, we have $\textbf{1}_n^T L = 0$. This implies that $\sum\limits_{i = 1}^n\sum\limits_{i \in N_i (t)}{a_{ij} (t)}\big\langle x_i(t) - x_0, x_j (t) - x_i (t)\big\rangle\leq0$. The following proof is similar to Theorem 3 and hence it is omitted.\end{proof}

\section{Discrete-time distributed algorithms for solving CFPs}\label{se5}

In this section, for discrete-time MAS (\ref{eq18}), the following input is presented to solve CFP (\ref{eq2}).
\begin{equation}\label{eq23}
\left\{ {\begin{array}{*{20}c}\begin{split}
   &{u_i (t) = h\sum\limits_{j \in N_i } {a_{ij} (x_j (t) - x_i (t)) + \phi _i (t)} }  \\
   &{\nabla _i (t) = \beta (t){\nabla g_i ^ +   (t)} } \\
   &{\xi _i (t) = x_i (t) + h\sum\limits_{j \in N_i } {a_{ij} (x_j (t) - x_i (t)) - \nabla _i (t)} }  \\
   &{\varphi _i (t) = \alpha (t)\left( {\xi _i (t) - P_{X_i } (\xi _i (t))} \right)}  \\
   &{\phi _i (t) =  - \nabla _i (t) - \varphi _i (t)}  \\
\end{split}\end{array}} \right.~~~~i\in\mathcal{V}
\end{equation}
where $\nabla g_i ^ +   (t)$  denotes the subgradient of function $g_i^+(y)$ at $y=x_i (t) + h\sum\limits_{j \in N_i } a_{ij} (x_j (t) - x_i (t))$, $h$ is the control gain to be designed. Note that each agent has only access to the information from its own inequality and set, as well as its own state and the relative states between itself and its neighbors, thus (\ref{eq23}) is distributed.
\begin{assumption}\label{A7}
$\nabla g^ + _i (\cdot)\leq K$ for some $K\geq 0$, $i=1, \cdots, n$.
\end{assumption}
\begin{lemma}\label{le12}
Given a linear system $ x(t+1) = Ax(t) + u(t)$, if the state matrix $A\in\mathbb{R}^{n\times n}$ is Schur stable and the control input $u(t)\in\mathbb{R}^{n} $ is such that $\lim\limits_{t\rightarrow\infty}u(t)=0$, then the linear system is asymptotically stable to zero, i.e., $\lim\limits_{t\rightarrow\infty}{x}(t)=0$.
\end{lemma}
\begin{proof}
It can be proved by the similar approach in Lemma \ref{le6} and using the fact that $\mathop {\lim }\limits_{k \to \infty } \sum\limits_{l = 0}^k {\rho ^{k - l} (A)\left\| {u(l)} \right\|}  = 0$ for $0<\rho(A)<1$, which has been proved in \cite{S.S. Ram 07}.
\end{proof}
The properties of graph's Laplacian matrix lead to the following lemmas directly \cite{C. Godsil 06}.
\begin{lemma}\label{le10}
For an undirected graph $\mathcal{G}(\mathcal{A})$, if $\mathcal{G}(\mathcal{A})$ is connected and $0<h<\frac{2}{\lambda_n}$, then it holds $\mathop {\max }\limits_{2 \le i \le n} \left| {1 - h\lambda _i (L)} \right|<1$.
\end{lemma}
\begin{lemma}\label{le11}
For a directed graph $\mathcal{G}(\mathcal{A})$, if $\mathcal{G}(\mathcal{A})$ has a spanning tree and $0<h<\mathop {\min }\limits_{2 \le i \le n} \frac{{2Re(\lambda _i (L))}}{{\left| {\lambda _i (L)} \right|^2 }}$, then it holds $\mathop {\max }\limits_{2 \le i \le n} \left| {1 - h\lambda _i (L)} \right|<1$.
\end{lemma}
\begin{proposition}
Suppose $\lim\limits_{t\rightarrow\infty}\phi _i(t)=0$ in $(\ref{eq23})$, $i\in\mathcal{V}$, if the undirected graph $\mathcal{G}(\mathcal{A})$ is connected and $0<h<\frac{2}{\lambda_n}$, then MAS (\ref{eq18}) with (\ref{eq23}) reaches consensus asymptotically.
\end{proposition}
\begin{proof}
Let $x(t)=\big[ x^T_1(t), \cdots, x^T_n(t) \big]^T$ and $\phi(t)=\big[ \phi^T_1(t), \cdots, \phi^T_n(t) \big]^T$, MAS (\ref{eq18}) with (\ref{eq23}) can be rewritten as
\begin{equation}\label{eq24}
{x}(t+1)=\left((I-hL )\otimes I_m\right) x(t)+\phi(t).
\end{equation}
Denote variable $\bar{x}(t) = \frac{1}{n}\sum\limits_{i = 1}^n x_i(t)=\frac{1}{n}\left( \textbf{1}_n^T \otimes I_m  \right)x(t)$.
Based on (\ref {eq24}), we have $\bar{x}(t+1)=\bar{x}(t)+u(t)$. Denote $e_i(t)=x_i(t)-\bar {x}(t)$ and  $e(t)=\big[ e^T_1(t), \cdots, e^T_n(t) \big]^T$. Note that if $e(t)\rightarrow0$ as $t\rightarrow\infty$, then MAS (\ref{eq18}) with (\ref{eq23}) reaches consensus asymptotically. From (\ref{eq24}), we have
\begin{equation}\label{eq25}\begin{split}
 e(t+1) =((I-hL) \otimes I_m )e(t) + \left(\left( I_n  - \frac{1}{n}\textbf{1}_n\textbf{1}_n^T \right) \otimes I_m  \right)\phi(t).
\end{split}
\end{equation}
Since $L$ is symmetric for $\mathcal{G}$ being undirected. We select $p_i\in\mathbb{R}^{n} $ such that $p_i^TL=\lambda_i(L)p_i^T$ and form an unitary
matrix $P=\left[\frac{\textbf{1}_n}{\sqrt{n}}, p_2, \cdots, p_n\right]$ to transform $I-hL$ into a diagonal form $diag(1, (1-h)\lambda_2(L), \cdots, (1-h)\lambda_n(L))=P^T(I-hL)P$. Denote $\tilde{e}(t)=P^Te(t)$ and partition $\tilde{e}(t)$ into two parts , i.e., $\tilde{e}(t)=[\tilde{e}_1^T(t), \tilde{e}_2^T(t)]^T$. Then, from (\ref{eq25}), we have
\[
\tilde{e}_1 (t + 1) = \left( {\left( {\frac{1}{{\sqrt n }}\textbf{1}_n^T \left( {I_n  - \frac{1}{n}\textbf{1}_n \textbf{1}_n^T } \right)} \right) \otimes I_m } \right)\phi(t).
\]
Note that ${\left( {\frac{1}{{\sqrt n }}\textbf{1}_n^T \left( {I_n  - \frac{1}{n}\textbf{1}_n \textbf{1}_n^T } \right)} \right) \otimes I_m } =0$ and $\tilde{e}_1 (t)=\frac{1}{{\sqrt n }}\left( {\textbf{1}_n^T  \otimes I_m } \right)e(t)=\frac{1}{{\sqrt n }}\sum\limits_{i = 1}^n {e_i (t)}=0$. Thus, it holds $\tilde e_1(t)=0$. Moreover, we have
\[
\tilde{e}_2 (t + 1) = \Lambda \tilde{e}_2 (t) + \left[ {\begin{array}{*{20}c}
   {\left( {p_2^T  - \frac{1}{n}p_2^T \textbf{1}_n \textbf{1}_n^T } \right) \otimes I_m }  \\
    \vdots   \\
   {\left( {p_n^T  - \frac{1}{n}p_n^T \textbf{1}_n \textbf{1}_n^T } \right) \otimes I_m }  \\
\end{array}} \right]\phi(t)
\]
where $\Lambda=diag((1-h\lambda_2(L))I_m,\cdots, (1-h\lambda_n(L))I_m )$. By Lemma \ref{le10}, we know if $0<h<\frac{2}{\lambda_n}$, $\Lambda$ is Schur stable. Recalling  Lemma \ref{le12} yields $\lim\limits_{t\rightarrow\infty}\tilde e_2(t)=0$. This and the fact that $\lim\limits_{t\rightarrow\infty}\tilde e_1(t)=0$ imply $\lim\limits_{t\rightarrow\infty} e(t)=0$, which leads to the validity of this result.
\end{proof}
\begin{proposition}
Suppose $\lim\limits_{t\rightarrow\infty}\phi _i(t)=0$ in $(\ref{eq23})$, $i\in\mathcal{V}$, if the directed graph $\mathcal{G}(\mathcal{A})$ has a spanning tree and $0<h<\mathop {\min }\limits_{2 \le i \le n} \frac{{2Re(\lambda _i (L))}}{{\left| {\lambda _i (L)} \right|^2 }}$, then MAS (\ref{eq18}) with (\ref{eq23}) reaches consensus asymptotically.
\end{proposition}
\begin{proof}
It can be proved by replacing the variable $\bar{x}(t)$ in the proof of Proposition 3 with $\hat {x}(t)$ defined in the proof of Proposition 1, and using the fact that $\mathop {\max }\limits_{2 \le i \le n} \left| {1 - h\lambda _i (L)} \right|<1$ if $\mathcal{G}$ has a spanning tree and $0<h<\mathop {\min }\limits_{2 \le i \le n} \frac{{2Re(\lambda _i (L))}}{{\left| {\lambda _i (L)} \right|^2 }}$, which is stated in Lemma \ref{le11}.
\end{proof}

 Now we give the convergence condition for (2) with (23) and its proof in detail when the graph is directed.
\begin{theorem}
Under Assumptions \ref{A1} and \ref{A7}, suppose $\{\alpha (t)\}$, $\{\beta (t)\}$ are two sequences such that

(a)  $\alpha(t) \in [0, 1]$, $\sum\limits_{t = 0}^\infty {\alpha (t)}  \to \infty$ and $\sum\limits_{t = 0}^\infty {\alpha^2 (t)}  < \infty$;

(b)  $0\leq\beta(t)\leq\infty$, $\sum\limits_{t = 0}^\infty {\beta(t)} \to \infty$ and $\sum\limits_{t = 0}^\infty {\beta^2 (t)} < \infty$.\\
If the directed graph $\mathcal{G}(\mathcal{A})$ is strongly connected and $0<h<\varrho$, where $\varrho=\min \Big[\frac{1}{\max\limits_{1 \le i \le n} \left(\sum\limits_{j = 1}^n a_{ij}\right)},$ $ \min\limits_{1 \le i \le n}\frac{2Re(\lambda _i (L))}{\left| {\lambda _i (L)} \right|^2 }\Big]$. Then, MAS (\ref{eq18}) with (\ref{eq23}) reaches consensus asymptotically, and the consensus state is in set $\textbf{X}^*$.
\end{theorem}
\begin{proof}
Since the graph is strongly connected, by Lemma \ref{le9}, there exists a vector $w = \left[ w_1  \cdots w_n \right]^T>0$ such that $w^TL=0$. Submitting (\ref{eq23}) to (\ref{eq18}), we have
\[
x_i(t+1)= \xi_i (t)-\varphi_i(t),~~~~i\in\mathcal{V}.
\]
Consider the positive-definite Lyapunov function candidate $V(t)=\sum\limits_{i = 1}^nw_i\|x_i(t)-x_0\|^2$, where $x_0\in\textbf{X}^*$. Taking the difference of function $V (t)$ yields
\[\begin{split}
 \Delta V(t)&=V(t+1)-V(t)\\
& =\sum\limits_{i = 1}^nw_i\|\xi_i (t) -\varphi_i (t)-x_0\|^2-\sum\limits_{i = 1}^nw_i\|x_i(t)-x_0\|^2\\
&=\sum\limits_{i = 1}^nw_i\|(1-\alpha(t))(\xi_i (t)-x_0)+\alpha(t)(P_{X_i} (\xi_i (t))-x_0)\|^2\\
&~~~-\sum\limits_{i = 1}^nw_i\|x_i(t)-x_0\|^2\\
&\leq\sum\limits_{i = 1}^nw_i\Big((1-\alpha(t))\|\xi_i (t)-x_0\|+\alpha(t)\|P_{X_i} (\xi_i (t))-x_0)\|\Big)^2\\
\end{split}
\]
\begin{equation}\label{eq26}
\begin{split}
&~~~-\sum\limits_{i = 1}^nw_i\|x_i(t)-x_0\|^2\\
&\leq\sum\limits_{i = 1}^nw_i\|\xi_i (t)-x_0\|^2-\sum\limits_{i = 1}^nw_i\|x_i(t)-x_0\|^2\\
&=\sum\limits_{i = 1}^nw_i \|y_i (t)-x_0\|^2 - \sum\limits_{i = 1}^nw_i\left\langle\nabla _i (t), y_i (t) - x_0 \right\rangle\\
&~~~+ \sum\limits_{i = 1}^nw_i\|\nabla _i (t)\|^2-\sum\limits_{i = 1}^nw_i\|x_i(t)-x_0)\|^2
\end{split}
\end{equation}
where $y_i(t)=x_i (t) + h\sum\limits_{j \in N_i } a_{ij} (x_j (t) - x_i (t))$ and the last inequality follows form using the non-expansiveness property of projection operator, i.e., $\|P_{X_i} (\xi_i (t))-x_0)\|\leq \|\xi_i (t)-x_0\|$. Since $\nabla g_i ^ +   (t)$ denotes the subgradient of function $g_i^+(y)$ at $y=y_i(t)$, we have
\begin{equation}\label{eq27}
 - \langle\nabla _i (t), y_i (t) - x_0 \rangle  \le  - \beta (t)g_i^ +  (y_i (t))\le 0.
\end{equation}
Moreover, since $0<h<\frac{1}{{\mathop {\max }\limits_{1 \le i \le n} \left( {\sum\limits_{j = 1}^n {a_{ij} } } \right)}}$, we have $0<1-h\sum\limits_{j = 1}^n {a_{ij}}<1$. By the convexity of the norm square function, it holds
\[\begin{split}
\|y_i (t)-x_0\|^2&=\|(1-h\sum\limits_{j \in N_i } l_{ij})(x_i (t)-x_0) + h\sum\limits_{j \in N_i } l_{ij} (x_j (t)-x_0 )\|^2\\
&\leq (1-h\sum\limits_{j \in N_i } l_{ij})\|x_i (t)-x_0\|^2+ h\sum\limits_{j \in N_i } l_{ij} \|x_j (t)-x_0 \|^2.
 \end{split}
\]
Thus, we have
\begin{equation}\label{eq28}
\begin{split}
\sum\limits_{i = 1}^nw_i \|y_i (t)-x_0\|^2&\leq\sum\limits_{i = 1}^nw_i  (1-h\sum\limits_{j \in N_i } l_{ij})\|x_i (t)-x_0\|^2\\
&~~~+ h\sum\limits_{i = 1}^nw_i \sum\limits_{j \in N_i } l_{ij} \|x_j (t)-x_0 \|^2\\
&=\sum\limits_{i = 1}^nw_i \|x_i (t)-x_0\|^2-h\sum\limits_{i = 1}^nw_i\left(\sum\limits_{j = 1}^nl_{ij}\right)\|x_i (t)-x_0\|^2\\
&~~~+ h\sum\limits_{j = 1}^n\left(\sum\limits_{i = 1}^nw_il_{ij}\right)\|x_j (t)-x_0\|^2\\
&=\sum\limits_{i = 1}^nw_i \|x_i (t)-x_0\|^2
 \end{split}
\end{equation}
where the last equation results from the fact that $\sum\limits_{j = 1}^nl_{ij}=0$ and $\sum\limits_{i = 1}^nw_il_{ij}=0$. Submitting (\ref{eq27}) and (\ref{eq28}) into (\ref{eq26}) yields
\begin{equation}\label{eq29}
\begin{split}
 \Delta V(t)\leq - \beta (t)\sum\limits_{i = 1}^nw_i {g_i^ +  (y_i (t)) }+ \sum\limits_{i = 1}^nw_i\|\nabla _i (t)\|^2.
\end{split}\end{equation}
From $(\ref{eq29})$, we have $V(t)\leq V(0)+\sum\limits_{t = 0}^{t-1} {\beta ^2 (t)(w^T\textbf{1}_n)K}\leq V(0)+\sum\limits_{t = 0}^\infty  {\beta ^2 (t)(w^T\textbf{1}_n)K}$ $ <\infty$. By the definition of $V(t)$, it can be concluded that $x_i(t)$ is bounded. By the fact that $\|\nabla_i(t)\|<\infty$, we know $\|\xi_i(t)\|<\infty$, this and the continuity of $P_{X_i} (\xi_i)$ imply $\|{\xi_i (t) - P_{X_i} (\xi_i (t))}\|<\infty$. Thus, $\lim\limits_{t\rightarrow\infty}\varphi_i(t)=0$. Since graph $\mathcal{G}(\mathcal{A})$ is strongly connected and $0<h<\mathop {\min }\limits_{2 \le i \le n} \frac{{2Re(\lambda _i (L))}}{{\left| {\lambda _i (L)} \right|^2 }}$, from Proposition 4, it can be concluded that MAS (\ref{eq18}) with (\ref{eq23}) reaches  consensus asymptotically, i.e., $\lim\limits_{t\rightarrow\infty}\|x_i(t)-x_j(t)\|=0$ for all $i,j \in \mathcal{V}$. Moreover, similar to (\ref{eq26}), we have
\[\begin{split}
 \Delta V(t)&=V(t+1)-V(t)\\
&=\sum\limits_{i = 1}^nw_i\|y_i (t)-x_0-\nabla _i (t)-\varphi_i (t)\|^2-\sum\limits_{i = 1}^nw_i\|x_i(t)-x_0)\|^2\\
&=\sum\limits_{i = 1}^nw_i\|y_i (t)-x_0-\nabla _i (t)-\varphi_i (t)\|^2-\sum\limits_{i = 1}^nw_i\|x_i(t)-x_0)\|^2\\
&=\sum\limits_{i = 1}^nw_i\|y_i (t)-x_0\|^2-2\sum\limits_{i = 1}^nw_i\left\langle\nabla _i (t)+\varphi_i (t), y_i (t)-x_0\right\rangle\\
&~~~+ \sum\limits_{i = 1}^nw_i\|\nabla _i (t)+\varphi_i (t)\|^2-\sum\limits_{i = 1}^nw_i\|x_i(t)-x_0)\|^2\\
&\leq-2\sum\limits_{i = 1}^nw_i\left\langle\nabla _i (t)+\varphi_i (t), y_i (t)-x_0\right\rangle+ \sum\limits_{i = 1}^nw_i\|\nabla _i (t)+\varphi_i (t)\|^2\\
&=-2\sum\limits_{i = 1}^nw_i\left\langle\nabla _i (t), y_i (t)-x_0\right\rangle-2\sum\limits_{i = 1}^nw_i\left\langle\varphi_i (t), \xi_i (t)-x_0\right\rangle \\
&~~~-2 \sum\limits_{i = 1}^nw_i\left\langle\varphi_i (t), \nabla _i (t)\right\rangle +\sum\limits_{i = 1}^nw_i\|\nabla _i (t)+\varphi_i (t)\|^2\\
&=-2\sum\limits_{i = 1}^nw_i\left\langle\nabla _i (t), y_i (t)-x_0\right\rangle-2\sum\limits_{i = 1}^nw_i\left\langle\varphi_i (t), \xi_i (t)-x_0\right\rangle\\
&~~~+\sum\limits_{i = 1}^nw_i\left(\|\nabla _i (t)\|^2+\|\varphi_i (t)\|^2\right)\\
\end{split}\]
\begin{equation}\label{eq30}
\begin{split}
&\leq - 2\beta (t)\sum\limits_{i = 1}^nw_i g^ +_i (x_i(t))-2\alpha(t)\sum\limits_{i = 1}^nw_i \|\xi_i(t)\|^2_{X_i}\\
&~~~+\sum\limits_{i = 1}^nw_i\left(\|\nabla _i (t)\|^2+\|\varphi_i (t)\|^2\right)
\end{split}
\end{equation}
where the first inequality results directly from (\ref{eq28}). Note that  $\sum\limits_{t = 0}^\infty \sum\limits_{i = 1}^nw_i\left(\|\nabla _i (t)\|^2+\|\varphi_i (t)\|^2\right)<\infty$ and $- 2\beta (t)\sum\limits_{i = 1}^nw_ig^ +_i (x_i(t))-2\alpha(t)\sum\limits_{i = 1}^nw_i \|\xi_i(t)\|^2_{X_i}<0$.
 By Lemma \ref{le8} and the fact that $\lim\limits_{t\rightarrow\infty} x_i(t)=\lim\limits_{t\rightarrow\infty} x_j(t)$, it can be concluded $V(t)$ converges and it holds $\sum\limits_{t = 0}^\infty (\beta (t)\sum\limits_{i = 1}^nw_i g^ +_i (x_i(t)) +\alpha(t)\sum\limits_{i = 1}^nw_i \|\xi_i(t)\|^2_{X_i})<\infty$. Since $ \beta (t) g_i^ +  (x_i(t))>0$ and $\alpha(t)\|\xi_i(t)\|^2_{X_i}>0$ for all $t>0$ and $i=1, \cdots, n$, we have $\sum\limits_{t = 0}^\infty \beta (t) g^ +_i  (x_i(t)) <\infty$ and  $\sum\limits_{t = 0}^\infty\alpha(t)\|\xi_i(t)\|^2_{X_i}<\infty$. By the facts $\sum\limits_{t = 0}^\infty {\alpha (t)}\rightarrow\infty$ and $\sum\limits_{t = 0}^\infty {\beta (t)}\rightarrow\infty$, we have $\lim\limits_{t\rightarrow\infty}\inf \left\| {\xi_i(t)}-P_{X_i}(\xi_i(t)) \right\|= \lim\limits_{t\rightarrow\infty}\inf g^ +_i(x_i(t))=0$. Thus, there exists a subsequence $\{x_i(t_k)\}$ of ${x_i(t)}$ such that $\lim\limits_{k\rightarrow\infty}x_i(t_{k})=x^*_i$, where $x^*_i$ is a vector such that $g^+_i(x^*_i)=h_i(x^*_i)=0$ for each $i=1, \cdots, n$. Recall the fact that $\lim\limits_{t\rightarrow\infty} x_i(t)=\lim\limits_{t\rightarrow\infty} x_j(t)$, we have $x^*_i=x^*_j$  for any $i, j\in \mathcal{V}$. Let $x^*=x^*_i$, we have $\lim\limits_{k\rightarrow\infty}x_i(t_{k})=x^*$. By the fact $\sum\limits_{i = 1}^nw_i\|x_i(t)-x_0\|^2$ converges and $\lim\limits_{t\rightarrow\infty} \dot{x}_i(t)=0$, we can conclude $\lim\limits_{t\rightarrow\infty} x_i(t)=\lim\limits_{k\rightarrow\infty}x_i(t_{k})=x^*$. Furthermore, note that $\nabla_i (t)\rightarrow 0$ as $t\rightarrow\infty$, thus $\lim\limits_{t\rightarrow\infty}\inf \left\| {\xi_i(t)}-P_{X_i}(\xi_i(t)) \right\|=0$ and $\lim\limits_{t\rightarrow\infty} x_i(t)=x^*$ imply $\lim\limits_{t\rightarrow\infty}\left\| {x^*}-P_{X_i}(x^*) \right\|=0$ for any $i\in\mathcal{V}$. This means ${x \in X= \cap _{i = 1}^n X_i }$. Therefore, $x^*$ is a feasible solution to CFP (\ref{eq2}), i.e., $x^*\in \textbf{X}^*$.
\end{proof}
\begin{corollary}
Under Assumptions \ref{A1} and \ref{A7}, suppose $\{\alpha (t)\}$, $\{\beta (t)\}$ are two sequences such that

(a)  $\alpha(t) \in [0, 1]$, $\sum\limits_{t = 0}^\infty {\alpha (t)}  \to \infty$ and $\sum\limits_{t = 0}^\infty {\alpha^2 (t)}  < \infty$;

(b)  $0\leq\beta(t)\leq\infty$, $\sum\limits_{t = 0}^\infty {\beta(t)} \to \infty$ and $\sum\limits_{t = 0}^\infty {\beta^2 (t)} < \infty$.\\
If the graph $\mathcal{G}(\mathcal{A})$ is undirected and strongly connected, $0<h<\frac{1}{ \max \limits_{1 \le i \le n} \sum\limits_{j = 1}^na_{ij} }$. Then, MAS (\ref{eq18}) with (\ref{eq23}) reaches consensus asymptotically, and the consensus state is in set $\textbf{X}^*$.
\end{corollary}
\begin{proof}
By Ger$\breve{\emph{\emph{s}}}$gorin Disc theorem, we can conclude $h<\frac{2}{\lambda_N}$ if $h<\frac{1}{ \max \limits_{1 \le i \le n} \sum\limits_{j = 1}^na_{ij} }$. Together with Lemma \ref{le10}, it can be proved by using the similar approach to Theorem 5 and hence the proof is omitted.
\end {proof}
\section{A special case: a distributed gradient-based algorithm for CFPs involving linear inequalities}\label{se6}

In this section, we will develop a distributed gradient-based algorithm for the CFP as follows.
\begin{equation}\label{eq31}
\left\{ {\begin{array}{*{20}c}\begin{split}
   &{A_{i}x-b_{i}\leq0}  \\
   &{x \in X:\mathop  = \limits^\Delta   \cap _{i = 1}^n X_i }  \\
\end{split}\end{array}} \right.\begin{array}{*{20}c}
   {} & {}  \\
\end{array}i = 1, \cdots ,n
\end{equation}
where $A_{i}\in\mathbb{R}^{m_{i}\times r}$ and $b\in\mathbb{R}^{m_{i}}$. It assumes CFP (\ref{eq31}) has a non-empty feasible solution set $\textbf{X}^*$.

For a vector $y=[y_1,\cdots, y_n]^T$, we define $y^+=[y_1^+,\cdots, y_n^+]^T$ and $y^-=[y_1^-,\cdots, y_n^-]^T$, where $y_i^+=\max (y_i,0)$ and $y_i^-=\min (y_i,0)$. We introduce a function $\psi(y)=\|y^+\|^2$. Note that $\psi(y)=0$ if and only if $y\leq0$. The function $\psi(y)$ is convex and differentiable. See the following lemma for detail.
\begin{lemma}\label{le13}
For any vector $y\in\mathbb{R}^{r}$, the function $\psi(y)=\|y^+\|^2$ is convex, differentiable and its gradient function at point $y$ is $\nabla_{y}\psi(y)=2y^{+}$.
\end{lemma}
\begin{proof}. For any vector $z\in\mathbb{R}^{r}$, we have $\psi(y+z)=\|(y+z)^{+}\|^{2}=\|y+z-(y+z)^{-}\|^{2}\leq\|y+z-(y)^{-}\|^{2} =\|y^{+}+z\|^{2}\leq\psi(y)+2[y^{+}]^{T}z+\|z\|^{2}$, where the first inequality follows from the fact that $(y+z)^{-} =\arg\min_{v\leq0}{\|(y+z)-v\|}$. Moreover, it holds that $\psi(y+z)=\|(y+z)-(y+z)^{-}\|^{2}=\|(y^{+}+[y^{-}+z-(y+z)^{-}]\|^{2}\geq\psi(y)+2[y^{+}]^{T}z+\|y^{-}+z-(y+z)^{-}\|^{2}\geq\psi(y)+2[y^{+}]^{T}z$, where the first inequality follows from the fact that it holds that  $[y^{+}]^{T}y^{-}=0$ and $[y^{+}]^{T}(y+z)^{-}\leq0$. Therefore, it holds that $\lim\limits_{\varepsilon\rightarrow0}\frac{\psi(y+\varepsilon\Delta y)-\psi y}{\varepsilon}=2[y^{+}]^{T}\Delta y$. This means $\nabla_{y}\psi(y)=2y^{+}$. From the fact that $\psi(y+z)\geq\psi(y)+2[y^{+}]^{T}z$, we know $\psi(y)$ is convex.
\end{proof}
Now we present the following distributed gradient-based algorithm for CFP (\ref{eq31}).

\begin{equation}\label{eq32}
\dot{x}_i(t)=\sum\limits_{i \in N_{i}} {a_{ij}} (x_j (t) - x_i (t))-\tau\left( A_i^T(A_ix_i(t)-b_i)^++x_i(t)-P_{X_i}(x_i(t))\right),~~~i=1,\cdots,n
\end{equation}
 where $\tau>0 $ is a positive coefficient, $x_{i}(t)\in\mathbb{R}^{r}$ represents the estimation value of the solutions to CFP (\ref{eq31}).
\begin{theorem}
If the graph $\mathcal{G}(\mathcal{A})$ is strongly connected, then $x_{i}(t)$ in (\ref{eq32}) converges to a fixed vector $x^*$ asymptotically for $i=1,\cdots,n$ and $x^*$ is in feasible solution set $\textbf{X}^*$ of (\ref{eq31}).
\end{theorem}
\begin{proof}
 By Lemma \ref{le13}, it is not difficult to prove that the term $A_i^T(A_ix_i-b_i)^+$ is the gradient of function $\|(A_ix_i-b_i)^+\|^2$. It can also be viewed as the unique subgradient of $\|(A_ix_i-b_i)^+\|^2$.  Then this result can be proved by the same method as Theorem 3 and hence it is omitted.
\end{proof}
\section{simulations}\label{se7}
In this section, we give numerical examples to illustrate the obtained results. Consider a multi-agent system consisting of five agents, the goal of the agents is to cooperatively search a feasibility $z^*=[z_1^*, z_2^*]^T$ of the CPF which includes two closed convex sets $X_1=\{(z_1,z_2)|2\leq z_1\leq4, 0\leq z_2\leq2\}$ and $X_2=\{(z_1,z_2)|2.5\leq z_1\leq4.5, 1\leq z_2\leq3\}$, and three linear inequalities $c (z)= 2z_1-3z_2-2 \leq 0$, $d (z) = 2z_1+3z_2-11 \leq 0$ and $q (z) = 8z_1-3z_2-28\leq0$. In Fig.1, the yellow region represents the feasible region. Set $X_i$ is only known to agent $i$ for $i=1, 2$, and agents 3, 4 and 5 can only have access to $c(z), d(z), q(z)$, respectively. In the following, we will present simulation results in three cases: The first two cases are for continuous-time distributed algorithms under the fixed and time-varying graphs, respectively. The third case is for the discrete-time distributed algorithm under the fixed graph. For each case, the communication graph is directed.

We first show the simulation result in the first case. The communication graph is shown in Fig. \ref{fig2}, which is strongly connected. The weight of each edge connecting different agents is 1. Set coefficient $\tau=20$ and let the initial state of each agent be $x_1(0)=[0, 5]^T, x_2(0)=[3, -2]^T, x_3(0)=[2, 3]^T, x_4(0)=[5, 1]^T, x_5(0)=[2, -3]^T$. The trajectory of MAS (\ref{eq1}) with (\ref{eq6}) is shown in Fig. \ref{fig3}. All agents also reach consensus at $z^*=[2.58, 1.23]^T$ which is a solution to the CFP. This is consistent with the result established in Theorem 3.

Now, we show the simulation result in the second case, the communication topologies switch between two bidirectional subgraphs depicted in Fig. \ref{fig4} and the switching law is given by Fig. \ref{fig5}. It is obvious that the $\delta-$graph associated with the time-varying graph is strongly connected. The weight of each edge connecting different agents is also being 1. Set coefficient $\tau=35$. Under the same initial condition as the first case, the trajectory of MAS (\ref{eq1}) with (\ref{eq6}) is shown in Fig. \ref{fig6}. All agents reach consensus at $z^*=[2.61, 1.37]^T$ while
remaining in the feasible region of the CFP. This is consistent with the result established in Theorem 4.

In addition, we show the simulation result in the third case. The communication topology in the first case is used to conduct this simulation. Set $\alpha(t)=\beta(t)=\frac{1}{0.02t+1}$ and $h=0.25$. Under the same initial condition as the last two cases, the trajectory of MAS (\ref{eq18}) with (\ref{eq23}) is shown in Fig. \ref{fig7}. All agents reach consensus at $z^*=[2.57, 1.54]^T$  which is a solution to the CFP. This accords with the result established in Theorem 5.

\begin{figure}
\begin{minipage}[t]{0.5\linewidth}
\centering
\includegraphics[width=0.85\textwidth]{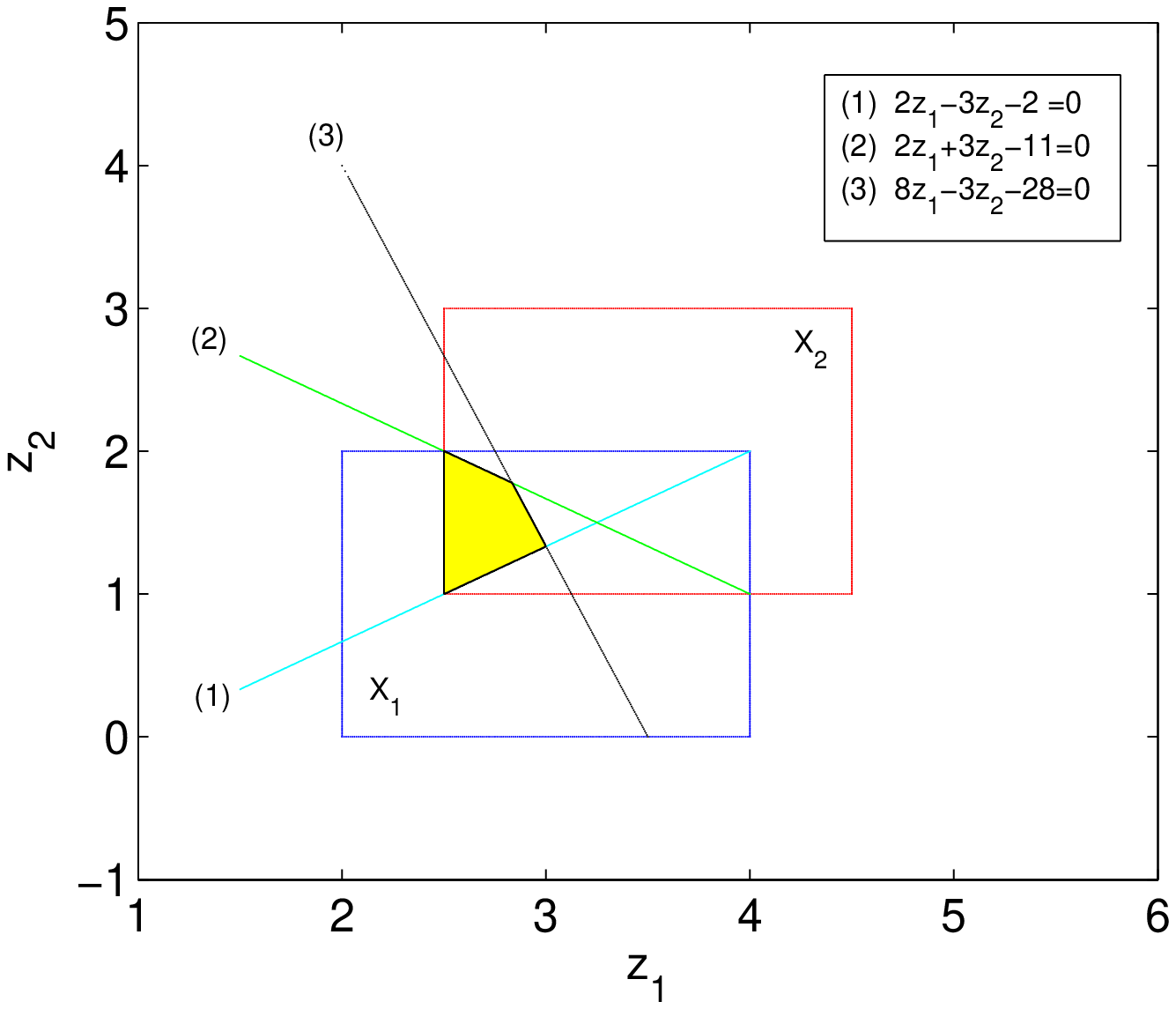}
\caption{ The feasible region of the CFP. }\label{fig1}
\end{minipage}
\begin{minipage}[t]{0.5\linewidth}
\centering
\includegraphics[width=0.4\textwidth]{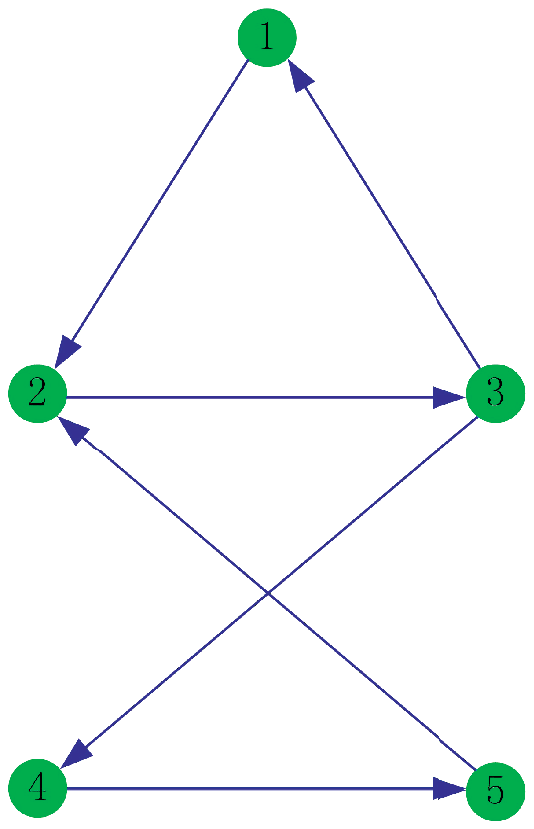}
\caption{The communication graph in the first case.}\label{fig2}
\end{minipage}
\end{figure}

\begin{figure}
\centering
\includegraphics[width=0.45\textwidth]{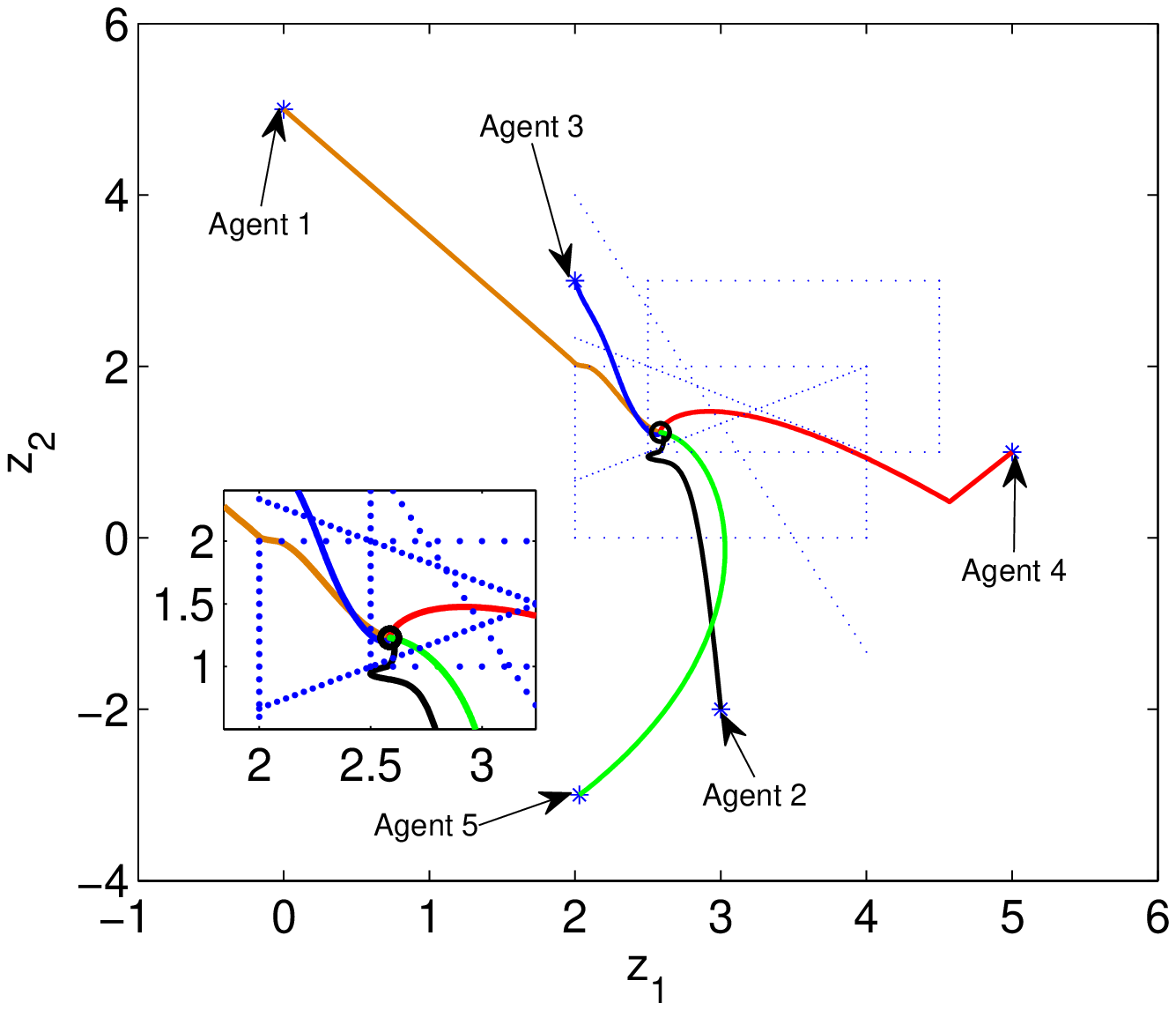}
\caption{The trajectory of the multi-agent system in the first case. Symbol ``*" represents the initial states of agents while ``$\circ$"  represents the final states of them.}\label{fig3}
\end{figure}

\begin{figure}[htbp]
\centering
\includegraphics[width=0.2\textwidth]{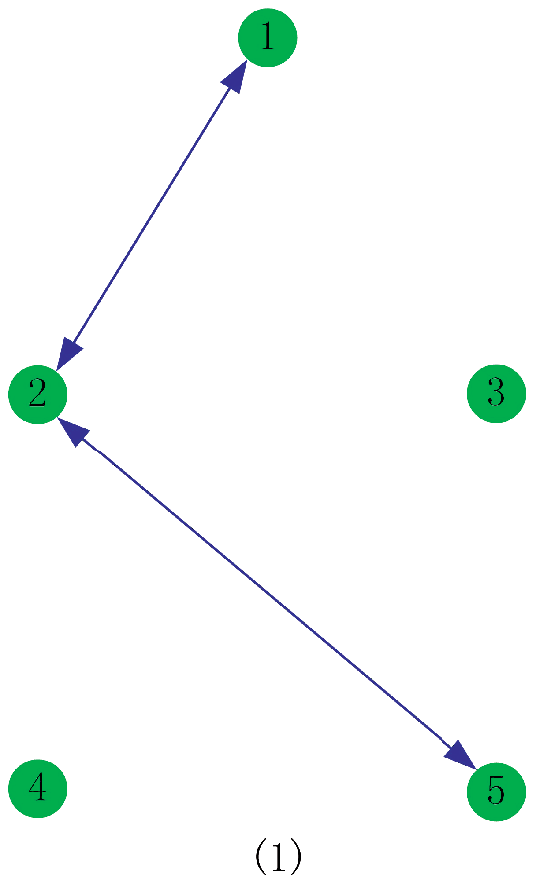}
~~~~~~~~~~~~~\includegraphics[width=0.2\textwidth]{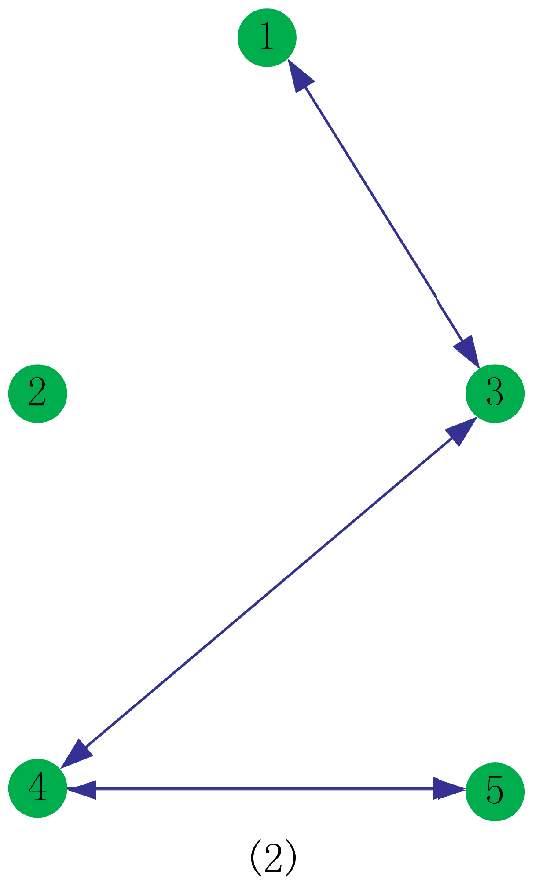}
\caption{The communication graph in the second case, which consists of  two subgraphs. The left one is labeled 1 and the right one is labeled 2.}\label{fig4}
\end{figure}

\begin{figure}
\begin{minipage}[t]{0.5\linewidth}
\centering
\includegraphics[width=0.85\textwidth]{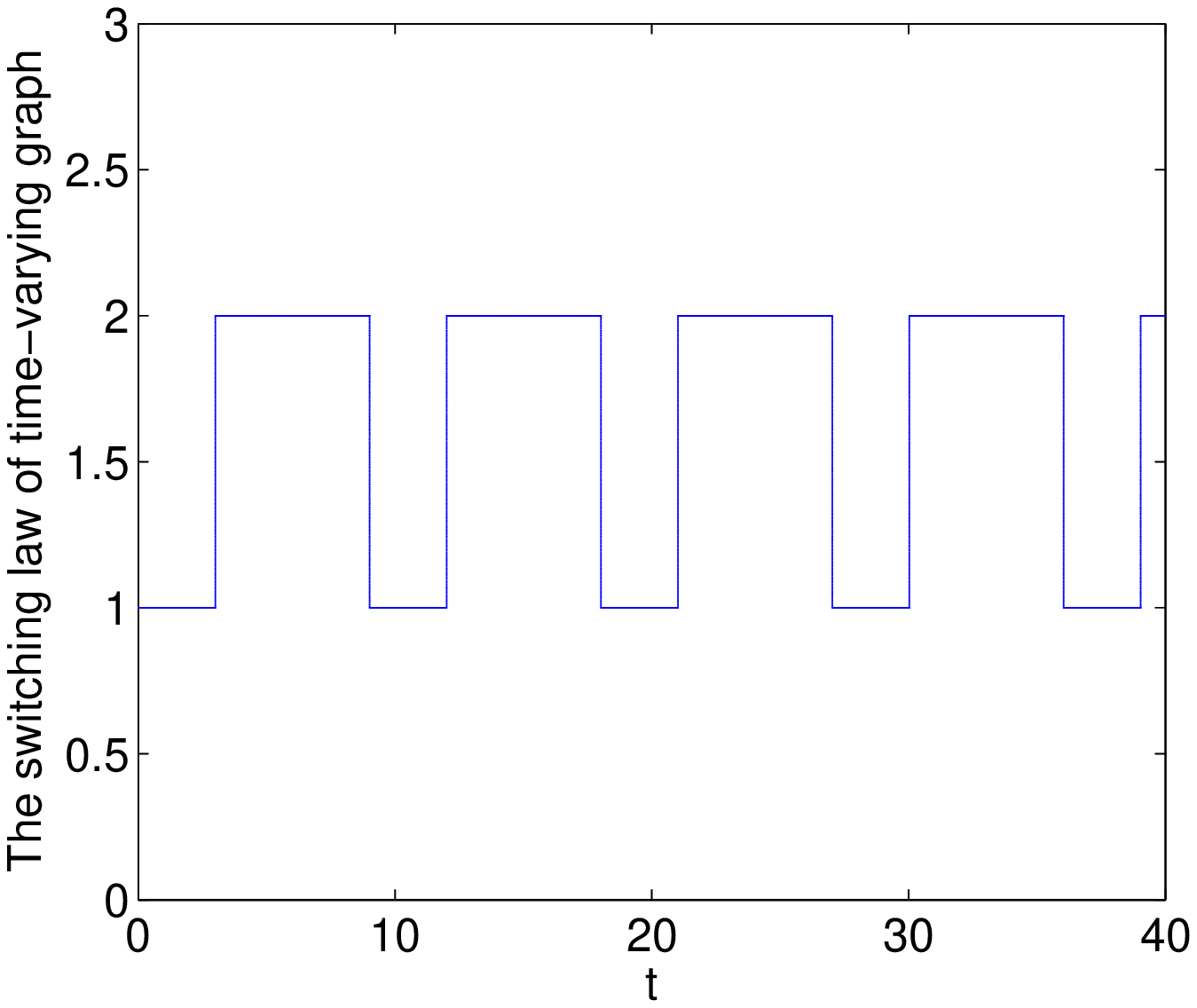}
\caption{ The switching law of the time-varying graph.}\label{fig5}
\end{minipage}
\begin{minipage}[t]{0.5\linewidth}
\centering
\includegraphics[width=0.85\textwidth]{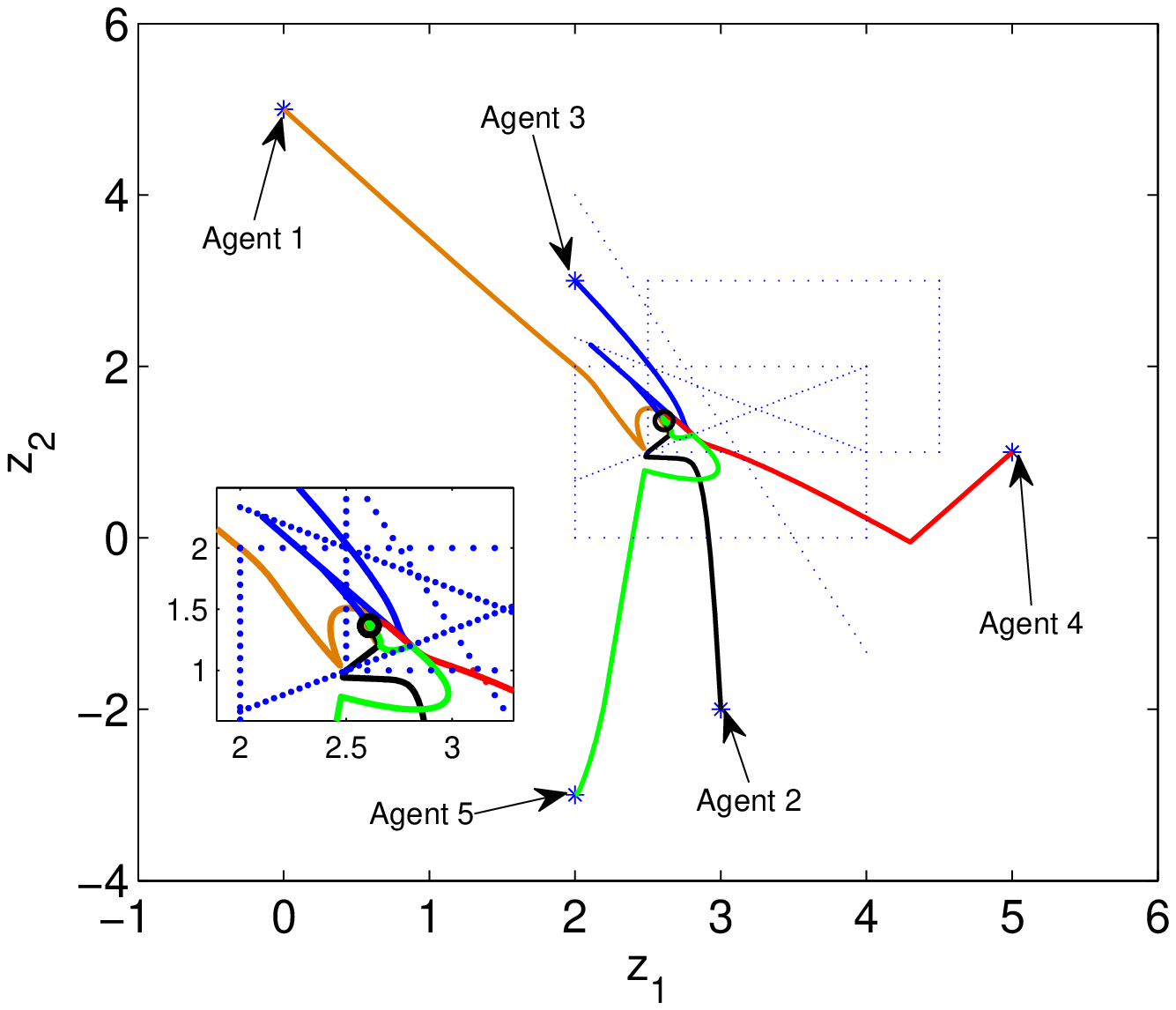}
\caption{The trajectory of the multi-agent system in the second case. Symbol ``*" represents the initial states of agents while ``$\circ$"  represents the final states of them.}\label{fig6}
\end{minipage}
\end{figure}

\begin{figure}
\centering
\includegraphics[width=0.45\textwidth]{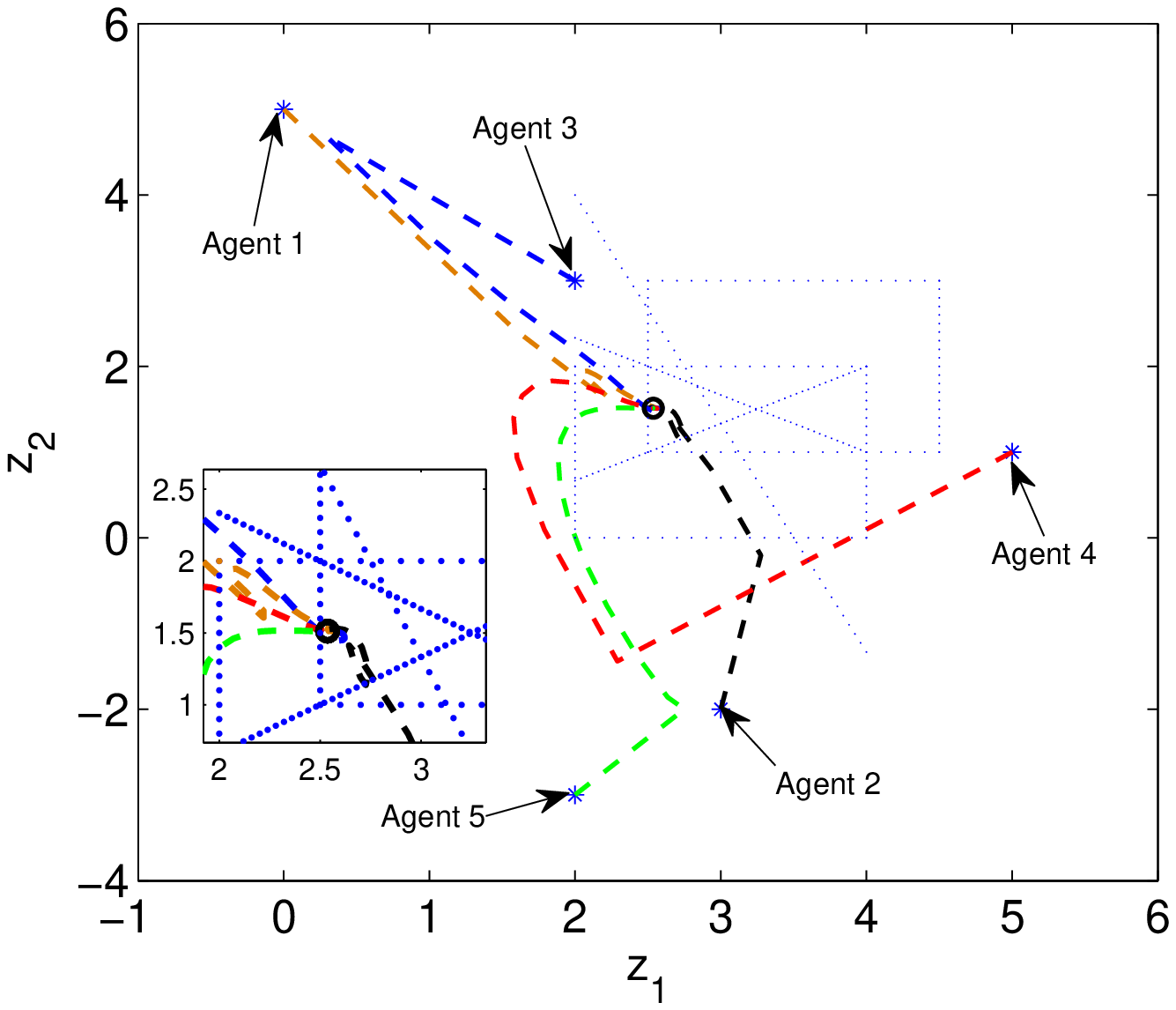}
\caption{The trajectory of the multi-agent system in the third case. Symbol ``*" represents the initial states of agents while ``$\circ$"  represents the final states of them.}\label{fig7}
\end{figure}
\section{Conclusions}\label{se8}

In this paper, the CFPs have been studied for multi-agent systems through local interactions. The distributed control algorithms were designed for both continuous- and discrete-time systems, respectively. In each case, a centralized approach was first introduced to solve the CFP. Then distributed control algorithms were proposed based on the subgradient and projection operations. The conditions associated with connectivity of the communication graph were given to ensure convergence of the distributed algorithms. The results showed that for the continuous-time case, if the communication graph is fixed and strongly connected, the MAS can reach consensus asymptotically and the consensus state is located in the solution set of the CFP. Moreover, the same result can be achieved if the $\delta-$graph associated with a time-varying graph is strongly connected. For the discrete-time case, under the condition of strong connectivity associated with the directed graph, if the control gain $h$ and the step-sizes $\alpha(t)$ and $\beta(t)$ are properly chosen, convergence of the distributed algorithm can also be guaranteed. Furthermore, a distributed gradient-based algorithm has been designed for a special case in which the CFP involves linear inequalities. Finally, simulation examples have been conducted to demonstrate the effectiveness of our results. Our future work will focus on some other interesting topics, such as the case under quantization, time delays, packet loss and communication bandwidth constraints, which will bring new challenges in solving CFPs over a network of agents.



%



\ifCLASSOPTIONcaptionsoff
  \newpage
\fi



%

\end{document}